\documentclass[11pt]{article}
\usepackage[T1]{fontenc}
\usepackage{amsfonts}
\usepackage{amsmath}
\usepackage{amssymb}
\usepackage{amsthm}
\usepackage{bbm}
\usepackage{bm}
\usepackage{mathrsfs}
\usepackage{verbatim}
\usepackage{setspace}
\usepackage{color}
\usepackage{pdfsync}
\usepackage{enumitem}
\usepackage{graphicx}
\usepackage{subfigure}
\usepackage{tikz}
\usetikzlibrary{patterns}
\usepackage{cases}
\usepackage{float}

\theoremstyle{plain}
\newtheorem{theorem}{Theorem}[section]
\newtheorem{proposition}[theorem]{Proposition}
\newtheorem{lemma}[theorem]{Lemma}
\newtheorem{corollary}[theorem]{Corollary}
\theoremstyle{definition}

\newtheorem{remark}[theorem]{Remark}

\theoremstyle{remark}

\renewenvironment{thebibliography}[1]{%
\begin{oldthebibliography}{#1}%
\setlength{\baselineskip}{1em}
\small
\setlength{\parskip}{0.4ex}%
\setlength{\itemsep}{.3em}%
}%
{%
\end{oldthebibliography}%
}

\newcommand{\eps}{\varepsilon}

\newcommand{\R}{\mathbb{R}}

\newcommand{\cF}{\mathcal{F}}

\newcommand{\cL}{\mathcal{L}}

\newcommand{\sL}{\mathscr{L}}

\DeclareMathOperator\tr{Tr}
\DeclareMathOperator{\Lip}{Lip}

\numberwithin{equation}{section}

\usepackage[pdfborder={0 0 0}]{hyperref}
\hypersetup{
  urlcolor = black,
  pdfauthor = {Johannes Muhle-Karbe, Marcel Nutz, Xiaowei Tan},
  pdfkeywords = {Equilibrium, Liquidity; Heterogeneous Beliefs},
  pdftitle = {Asset Pricing, Heterogeneous Beliefs and Liquidity},
  pdfsubject = {Asset Pricing, Heterogeneous Beliefs and Liquidity},
  pdfpagemode = UseNone
}

\begin{document}

\title{\vspace{-3em}
Asset Pricing with Heterogeneous Beliefs\\and Illiquidity\footnote{The authors are grateful to the Associate Editor and two anonymous referees for their helpful comments.}
\date{\today}
\author{
  Johannes Muhle-Karbe%
  \thanks{Department of Mathematics, Imperial College London, j.muhle-karbe@imperial.ac.uk. Research supported by the CFM-Imperial Institute of Quantitative Finance.}
  \and
  Marcel Nutz%
  \thanks{
  Departments of Statistics and Mathematics, Columbia University, mnutz@columbia.edu. Research supported by an Alfred P.\ Sloan Fellowship and NSF Grants DMS-1512900 and DMS-1812661.
  }
  \and
  Xiaowei Tan%
  \thanks{
  Department of Mathematics, Columbia University, xt2161@columbia.edu.
  }  
 }
}
\maketitle \vspace{-1.2em}

\begin{abstract}
  This paper studies the equilibrium price of an asset that is traded in continuous time between $N$ agents who have heterogeneous beliefs about the state process underlying the asset's payoff. We propose a tractable model where agents maximize expected returns under quadratic costs on inventories and trading rates. The unique equilibrium price is characterized by a weakly coupled system of linear parabolic equations which shows that holding and liquidity costs play dual roles. We derive the leading-order asymptotics for small transaction and holding costs which give further insight into the equilibrium and the consequences of illiquidity.
\end{abstract}

\vspace{1em}

{\small

}
\vspace{1em}

\section{Introduction}\label{se:intro}

Heterogeneous beliefs about fundamental values are a key motive for trade in financial markets. Accordingly, a rich literature studies how prices form as the aggregate of subjective beliefs; see e.g.\ the survey \cite{ScheinkmanXiong.04} for numerous references. This synthesis happens by means of trading: agents with lower individual valuations sell to agents who are more optimistic about fundamentals. Hence, \emph{liquidity}---the ease with which trades can be implemented---plays an important role in determining how beliefs are reflected in prices.

In the present study, we propose a tractable model that allows us to study the interplay of heterogeneous beliefs and liquidity in determining asset prices. We consider $N$ (types of) agents who have different beliefs about the state process determining the payoff of a given asset. They trade the asset in continuous time to maximize their expected returns, penalized with quadratic costs on inventories and trading rates. We show that this model admits a unique Markovian equilibrium. The equilibrium price is characterized as the solution of a \emph{linear} system of parabolic equations with a weak coupling (i.e., the equations are coupled only through the zeroth-order terms). The solution, as well as the necessary estimates on its derivatives, are obtained by combining a fixed-point argument of~\cite{BechererSchweizer.05} for reaction--diffusion equations, classical Schauder theory for parabolic equations and a gradient estimate that seems to be novel.

This characterization allows us to study the influence of the two costs. The holding costs on inventories, parametrized by a coefficient $\gamma$, can be seen as a proxy for risk aversion, whereas the costs on the trading rate, with coefficient $\lambda>0$, stand in for the liquidity (or transaction) cost caused by market impact. The two costs determine how agents take into account current and future expected returns when choosing their portfolios: with bigger transaction costs, further weight is placed on future market conditions to avoid trades likely to be reversed later on. Conversely, larger holding costs make it less appealing to hold a given position, so that the current trading opportunities play a bigger role. Accordingly, liquidity and holding costs play inverse roles in our analysis. Specifically, when the asset is in zero net supply (a natural assumption for derivative contracts, say) the two costs only enter through their ratio $\gamma/\lambda$. For a positive supply, the asset price remains invariant if the inverse of the supply is rescaled in the same manner as transaction and holding costs, so that the larger trading and holding costs of bigger asset positions are offset by reducing both frictions.

Explicit asymptotic formulas obtain in the limiting regimes where either transaction costs or holding costs are small ($\gamma/\lambda \approx \infty$ or $\gamma/\lambda \approx 0$, respectively). For small transaction costs $\lambda \to 0$,  a singular perturbation expansion identifies the leading-order correction term relative to the frictionless equilibrium in which assets are priced by taking conditional expectations under a representative agent's probability measure that averages the agents' beliefs. 

The correction term turns out to be proportional to the square root $\sqrt{\lambda}$ of the transaction costs. The corresponding constant of proportionality is related to the average of the subjective drifts of the agents' frictionless portfolios. Thus, equilibrium prices increase relative to their frictionless counterparts if agents on average expect to increase their positions in the future, and vice versa. The interpretation is that in illiquid markets, agents take into account their future trading needs to reduce  transaction costs. Accordingly, expectations of future purchases already lead to increased positions earlier on and equilibrium prices increase according to the excess demand created by the aggregated adjustments of all agents, and vice versa.

The equilibrium for small holding costs $\gamma \to 0$ can be approximated by a regular perturbation expansion around the risk-neutral equilibrium price which averages all agents' subjective conditional expectations. Here, the leading-order correction term is determined by $\gamma$ times the average of the agents' expectations of their future positions. Other things equal, agents reduce the magnitude of their positions when holding costs are introduced, thereby reducing the demand of agents who expect to be long on average and increasing the demand of agents who expect to be short. The resulting sign of the price correction therefore depends on the aggregate expectations in the market. 

To illustrate the implications of these results and test the accuracy of the expansions, we consider an example where the state process determining the asset's payoff has Ornstein--Uhlenbeck dynamics, a simple model for a forward contract on a mean-reverting underlying such as an FX rate. Agents agree on the mean-reversion level and volatility, but disagree about the speed of mean-reversion. For these linear state dynamics, the parabolic PDE system describing the equilibrium price can be reduced to a system of linear ODEs by a suitable ansatz, and in turn compared to the explicit formulas that obtain for our small-cost asymptotics in this case. In this example, we find that the introduction of small transaction costs increases volatility, in line with the asymmetric information model of~\cite{danilova.juillard.19}, the risk-sharing model studied in~\cite{herdegen.al.19}, numerical results of~\cite{adam.al.15,buss.al.16} and empirical studies such as~\cite{hau.06,jones.seguin.97,umlauf.93}. By contrast, the introduction of small holding costs decreases the equilibrium volatility. The reason is the opposite manner in which the two costs influence how agents take into account future trading opportunities. Without transaction costs, agents who believe in faster than average mean-reversion perceive a mean-reverting price process and therefore sell when its value is high, whereas agents who believe in slower mean-reversion perceive a price process that exhibits ``momentum'' and therefore buy in this case. While this frictionless tradeoff only depends on the current dynamics of the asset, transaction costs force the agents to take into account future trading opportunities as well. That makes the current trading opportunities less attractive for the agent believing in faster mean-reversion and therefore creates an excess demand for the asset when its price is high. This in turn further increases high prices and conversely  decreases low ones, leading to additional volatility.

Holding costs have the opposite effect, by discounting the importance of future trading opportunities and therefore reducing volatility relative to the risk-neutral limiting price. In fact, the exact equilibrium volatility smoothly interpolates between the risk-neutral volatility (which is highest) and its counterpart without transaction costs (which is lowest).
For model parameters estimated from time series data for the USD/EUR exchange rate, we find that the exact equilibrium prices agree with these comparative statics gleaned from their asymptotic approximations. 


For models where trading is frictionless, there is an extensive literature on asset pricing under heterogeneous beliefs; see, e.g., \cite{BhamraUppal.14, DetempleMurthy.94,ScheinkmanXiong.04,cvitanic.al.11}
and the references therein. To obtain tractable results with limited liquidity, we focus on a model with quadratic holding and trading costs as well as linear preferences over gains and losses. 

Similar linear-quadratic liquidity models are used in partial equilibrium contexts by~\cite{almgren.12,annkirchner.kruse.15,bank.al.17,kohlmann.tang.02}.  Risk-sharing equilibria with homogeneous beliefs are studied in \cite{BouchardFukasawaHerdegenKarbe.18,garleanu.pedersen.16,herdegen.al.19,sannikov.skrzypacz.16}.  As the corresponding first-order conditions are linear, these models are considerably more tractable than equilibrium models with other preferences or trading costs, where analytical results are only available if prices or trading strategies are deterministic~\cite{lo.al.04,vayanos.98,vayanos.vila.99,weston.17} or agents only trade once~\cite{ScheinkmanXiong.03,davila.17}. Numerical analyses of equilibrium models with heterogeneous beliefs and transaction costs are carried out in~\cite{adam.al.15,buss.al.16}.

Considering a holding cost on risky positions as in \cite{cartea.jaimungal.16,choi.al.18,NutzScheinkman.17,sannikov.skrzypacz.16} further simplifies the analysis compared to models where the corresponding risk penalty is imposed on the variance of the risky positions as in~\cite{garleanu.pedersen.13,garleanu.pedersen.16,herdegen.al.19}. Indeed, in the present model, we can characterize the equilibrium price by a system of linear PDEs, avoiding the nonlinear equations that naturally appear in models where agents have risk aversion in the form of concave utility functions.
As the present work focuses on equilibrium asset prices with heterogeneous \emph{beliefs} about the underlying state process, we abstract from heterogeneous holding costs. These are considered in \cite{herdegen.al.19} for agents with homogeneous beliefs and in~\cite{BayraktarMunk.17,casgrain.jaimungal.18} for partial equilibrium models with heterogeneous beliefs.

This paper is organized as follows. Section~\ref{se:model} details the financial market and the definition of an equilibrium. In Section~\ref{se:singleAgentOpt} we derive the optimal portfolio of any agent given an exogenous asset price process. Section~\ref{se:equilibrium} provides the existence, uniqueness and PDE characterization of the equilibrium price. The leading-order asymptotics for small transaction and holding costs are presented in Sections~\ref{se:transactionAsymp} and~\ref{se:holdingAsymps}, respectively. The concluding Section~\ref{se:example} covers the example with mean-reverting state process.

\paragraph{Notation.}

As usual, $C=C(\R^{n})$ is the space of continuous functions $g(x)$ on $\R^{n}$ and $C^{k}$ is the space of functions $g\in C$ whose partial derivatives up to order $k$ exist and belong to $C$. Similarly, $C^{1,k}$ is the space of continuous functions $g(t,x)$ such that $g(t,\cdot), \partial_{t}g(t,\cdot)\in C^{k}$. For any of these spaces, a subscript ``$b$'' indicates that the functions \emph{and} all mentioned derivatives are bounded. The dimension~$n$ of the underlying domain is often understood from the context. Conditional expectations are denoted $E_t[\cdot]=E[\cdot|\cF_t]$ for brevity and when $F$ is a functional of the paths of a process~$Y$, we will often write $E_{t,y}[F(Y)]=E[F(Y)|Y_t=y]$. In this context, $Y$ will be the solution of an SDE and $E[F(Y)|Y_t=y]$ can be unambiguously defined as $E[F(Y^{t,y})]$ where $(Y^{t,y}_{s})_{s\geq t}$ is the unique solution of the corresponding SDE with initial condition~$y$ at time~$t$. 


\section{Model}\label{se:model}

\paragraph{Beliefs.}

Let $X$ be the coordinate-mapping process on the space $\Omega=C_{0}([0,T],\R^d)$ of continuous, $d$-dimensional paths $\omega$ with $\omega_{0}=0$, equipped with the canonical $\sigma$-field~$\cF$ and filtration $(\cF_{t})_{t \in [0,T]}$ generated by $X$. We consider $N$ (types of) agents with heterogeneous views on the distribution of the state process~$X$. Specifically, for each $1\leq i\leq N$, let $Q_{i}$ be a probability measure on $\Omega$ under which~$X$ satisfies
\begin{align}
\label{eq:Xdynamics}
dX_t=b_i(t,X_t) dt+\sigma_i(t,X_t) dW^i_t
\end{align}
where $W^i$ is a $d'$-dimensional Brownian motion. We assume that $b_i: [0,T]\times\R^d\to\R^d$ and $\sigma_i: [0,T]\times\R^d\to\R^{d\times d'}$ are jointly Lipschitz and bounded. This guarantees in particular that~\eqref{eq:Xdynamics} has a unique strong solution. Moreover, we assume that the matrix $\sigma^{2}_{i}:=\sigma_i\sigma_i^\top$ is uniformly parabolic: there is a constant $\kappa>0$ such that $\xi^{\top}\sigma^{2}_{i}\xi\geq \kappa |\xi|^{2}$ for all $\xi\in\R^{d}$. The associated generator is denoted by
\begin{equation}\label{eq:Li}
  \cL^{i}=\partial_{t} + b_i\partial_{x} + \frac{1}{2}\tr\sigma_i^2\partial_{xx}.
\end{equation}
We allow the measures $Q_{i}$ to differ in drift as well as volatility. The differences regarding the drifts are more important in practice, and our results remain relevant also without disagreement on volatility. We refer to \cite{EpsteinJi.2013} for a detailed discussion of volatility uncertainty.

\begin{remark}\label{rk:support}
  The above assumptions imply that the support of $Q_{i}$ is the whole space~$\Omega$; cf.\ \cite[Theorem~3.1]{StroockVaradhan.72}. Thus, if $F$ and $G$ are continuous functions on~$\Omega$, then $F(X)=G(X)$ $Q_{i}$-a.s.\ is equivalent to $F=G$. This fact will be used throughout the paper, often implicitly. Uniform parabolicity is convenient to simplify the exposition, but of course results similar to ours could be obtained under different assumptions. When the support of $X$ is not the whole space, the statements involving price functions and PDEs need to be restricted to a suitable domain (as, e.g., in \cite{MuhleKarbeNutz.16}).
\end{remark}

\paragraph{Market Model.}

Let $f\in C^{3}_{b}(\R^{d})$.
 We consider $N$ agents that dynamically trade an asset with a single payoff $f(X_{T})$ at the time horizon $T>0$. Fix a constant $a_{0}\geq0$, the \emph{exogenous supply} at time $t=0$, and the initial asset allocation $a_{i}\in\R$ to each agent, where $\sum_{i=1}^N a_{i}=a_{0}$. 
Let $\sL^{p}(Q_{i})$ denote the set of progressively measurable processes $\phi=(\phi_{t})_{0\leq t\leq T}$ (of appropriate dimension) such that $E^{i}[\int_{0}^{T} \phi^{p}_{t} dt]<\infty$. An (admissible) \emph{portfolio} for agent~$i$ is a scalar process $\phi\in\sL^{4}(Q_{i})$ which satisfies $\phi_0=a_{i}$ and is absolutely continuous with rate $\dot\phi\in\sL^{4}(Q_{i})$.\footnote{The precise integrability condition is not crucial; we simply need to ensure that the local martingale part of $\int \phi dS$ has vanishing expectation when $S$ is defined as in~\eqref{eq:SunderQi}. In our main equilibrium result the optimal portfolios are bounded.} We say that portfolios $\phi^{i}$, $1\leq i\leq N$ \emph{clear the market} if 
$$
\sum_{i=1}^N\phi^i_t=a_{0},\quad t\in [0,T]
$$
holds pointwise. A \emph{price process} (for $f$) is a progressively measurable process $S=(S_{t})_{0\leq t\leq T}$ which satisfies $S_{T}=f(X_{T})$ and is an It\^o process with sufficiently integrable coefficients under each $Q_{i}$: 
\begin{equation}\label{eq:SunderQi}
  dS_{t}=\mu^i_t dt+\nu^i_t dW^i_t \quad\mbox{with}\quad \mu^{i},\nu^{i}\in\sL^{4}(Q_{i})
\end{equation}
for some $Q_{i}$-Brownian motion $W^{i}$, for all $1\leq i\leq N$.

\paragraph{Equilibrium.}

To formulate the agents' optimization criteria, we fix a \emph{holding cost} parameter $\gamma>0$ and a \emph{transaction cost} parameter $\lambda>0$.
(The boundary cases $\lambda=0$ and $\gamma=0$ will be considered in Sections~\ref{se:transactionAsymp} and~\ref{se:holdingAsymps}, respectively.)
For a given price process~$S$, agent $i$ maximizes her expected returns, penalized for inventories and trading costs,
\begin{equation}\label{eq:goal}
  J^{i}(\phi) =E^i\left[\int_0^T \left(\phi_t dS_{t} - \frac{\gamma}{2} \phi_t^2 dt -\frac{\lambda}{2}\dot{\phi}_t^2 dt\right)\right]
\end{equation}
over the set of her admissible portfolios. A portfolio $\phi^{i}$ is \emph{optimal} for agent~$i$ if it is a maximizer.  If $S$ is a price process such that there exist optimal portfolios~$\phi^{i}$, $1\leq i\leq N$ for the agents which clear the market, then $S$ is an \emph{equilibrium price process}. Finally, $v: [0,T]\times\R^{d}\to\R$ is an \emph{equilibrium price function} if $v(t,X_{t})$ defines an equilibrium price process. We shall be interested in symmetric equilibria with prices of this Markovian form; however, the associated portfolios are usually path-dependent in the presence of transaction costs. As is implicit in~\eqref{eq:goal}, the interest rate is assumed to be zero throughout.

In this model, agents have fixed beliefs $Q_{i}$ and agree-to-disagree. The beliefs can be inconsistent with one another or with observations of $X$ over time, especially in the case of disagreement on future volatilities. As in \cite{HarrisonKreps.78}, the beliefs are used by the agents to compute \emph{expected} future profits or losses and eventually determine the initial price $S_{0}$ of the asset. This occurs at the initial time and without actual observation of the future. Thus, the resulting price is consistent for all agents even though they disagree.

The linear-quadratic criterion~\eqref{eq:goal} includes several simplifications to enhance the tractability of the model.
First, as in \cite{heaton.lucas.96}, transaction costs only depend on each agent's individual trading rate and not on the total order flow in the market. Accordingly, the trading cost should be interpreted as a tax or the fee charged by an exchange, rather than as a price impact cost as in \cite{almgren.chriss.01}. Partial equilibrium models where the agents interact through their common price impact are studied in \cite{bonelli.al.18}, for example. Nevertheless, to obtain tractable first-order conditions, we use a quadratic rather than linear cost. This simplification is motivated by recent results~\cite{gonon.al.19} for risk-sharing equilibria
which show that equilibrium prices are robust with respect to the specification of the trading cost. Second, we assume as in \cite{sannikov.skrzypacz.16} that all investors penalize inventories through a quadratic holding cost on positions. This leads to a system of \emph{linear} PDEs for equilibrium prices, unlike the systems of nonlinear PDEs that appear for holding costs on variances in~\cite{herdegen.al.19} even with homogeneous beliefs. Concave utility functions over terminal wealth or intermediate consumption would further complicate the analysis by introducing each agent's value function as an additional component of the nonlinear PDE system. Third, we suppose that the inventory costs are homogeneous across agents in order to single out the effect of heterogeneity in beliefs. Heterogeneous holding costs have been studied in~\cite{BouchardFukasawaHerdegenKarbe.18,herdegen.al.19} for models with homogeneous beliefs. Finally, we do not model ``where the transaction costs go'' in equilibrium: like in most of the related literature~\cite{buss.dumas.19, lo.al.04,vayanos.98}, we consider the trading cost as a deadweight loss for the financial market under consideration. This seems reasonable for a transaction tax or the fees imposed by an exchange, for example.

\section{Single-Agent Optimality}\label{se:singleAgentOpt}

As a preparation for the equilibrium result, we first fix agent~$i$ and solve her individual optimization problem in the face of an exogenous price process. Similar linear-quadratic optimization problems have been considered, e.g., in \cite{annkirchner.kruse.15,bank.al.17,BouchardFukasawaHerdegenKarbe.18,cartea.jaimungal.16,garleanu.pedersen.13,garleanu.pedersen.16,kohlmann.tang.02}; we provide a  self-contained derivation for the convenience of the reader. 

\begin{lemma}\label{le:optPortfolio}
  Let $\gamma>0$, $\lambda>0$ and 
  \begin{align}
  \label{eq:defG}
  G(t)=\cosh\left(\sqrt{\frac{\gamma}{\lambda}}(T-t)\right), \quad t \in [0,T].
  \end{align}
  Given a price process~$S$ as in~\eqref{eq:SunderQi}, the $dt\times Q_{i}$-a.e.\ unique optimal portfolio for agent~$i$ is 
  \begin{equation}\label{eq:position}
  \phi^i_t=\frac{G(t)}{G(0)}a_i + \int_0^t\frac{G(t)}{G(s)}E^i_s\left[\int_s^T\frac{G(u)}{G(s)}\frac{\mu^i_u}{\lambda}\,du\right]ds.
  \end{equation}
  In particular, the optimal trading rate is characterized by the random ODE
  \begin{equation}\label{eq:rate}
  \dot{\phi}^i_t = \frac{G'(t)}{G(t)}\phi^i_t+ E^i_t\left[ \int_t^T \frac{G(s)}{G(t)}\frac{\mu^i_s}{\lambda}\,ds\right], \quad \phi^i_0=a_i.
  \end{equation}
\end{lemma}

As in the previous literature, the optimal trading strategy tracks the average $E^i_t[ \int_t^T -\frac{\gamma G(s)}{\lambda G'(t)}\frac{\mu^i_s}{\gamma}\,ds]$ of the discounted future values of the no-transaction cost portfolio $\mu^i_t/\gamma$ obtained by pointwise maximization of the drift of~\eqref{eq:goal}. To wit, illiquidity is accounted for by ``aiming in front of the moving target''~\cite{garleanu.pedersen.13}. Both the tracking speed $-G'(t)/G(t)$ and the discount kernel $K(t,s)=-\gamma G(s)/\lambda G'(t)$ are determined by the ratio $\gamma/\lambda$ of holding and transaction costs, with relatively lower transaction costs leading to faster trading and more emphasis on the current returns of the asset.
 
\begin{proof}[Proof of Lemma~\ref{le:optPortfolio}]
    Direct differentiation shows that $\dot{\phi}^i$ of~\eqref{eq:rate} is indeed the derivative of $\phi^i$ in~\eqref{eq:position}. Moreover, $\mu^{i}\in \sL^{4}(Q_{i})$ and Doob's inequality imply that $\phi^{i}\in \sL^{4}(Q_{i})$ and then~\eqref{eq:rate} yields that $\dot\phi^{i}\in \sL^{4}(Q_{i})$.
  
  Note that
  $$
    J^{i}(\phi)=E^i\left[\int_0^T \left(\phi_t \mu^i_t - \frac{\gamma}{2} \phi_t^2 -\frac{\lambda}{2}\dot{\phi}_t^2\right)dt\right]
  $$
  for any portfolio $\phi$ and that portfolios can be parametrized by their trading rates as the initial allocations are fixed and $\dot\phi\in\sL^{4}(Q_{i})$ implies $\phi\in\sL^{4}(Q_{i})$. The strict concavity of~$J^{i}$ implies that any optimizer is (a.e.) unique and that a trading rate $\dot{\phi}$ is optimal if and only if the G\^ateaux derivative $\lim_{\eps \to 0} \frac{1}{\eps}[J^{i}(\dot{\phi}+\eps \dot{\vartheta})-J^{i}(\dot{\phi})]$ of~\eqref{eq:goal} vanishes in all directions $\dot{\vartheta} \in \sL^{4}(Q_i)$; that is, 
  \begin{align*}
  0 &= E^i\left[ \int_0^T\left( \mu^i_t \int_0^t \dot\vartheta_s ds -\gamma \phi^i_t \int_0^t \dot\vartheta_s ds -\lambda \dot{\phi}^i_t \dot{\vartheta}_t \right)dt\right]\\
  & =E^i\left[\int_0^T \left(\int_t^T \left(\mu^i_s -\gamma \phi^i_s \right) ds -\lambda \dot{\phi}^i_t\right) \dot{\vartheta}_t dt\right], \quad \dot{\vartheta} \in \sL^{4}(Q_i).
  \end{align*}
  As $\dot{\vartheta}$ is arbitrary, this is equivalent to $\dot{\phi}^i_t =\frac{1}{\lambda} E^i_t [\int_t^T (\mu^i_s -\gamma \phi^i_s ) ds]$, which is in turn equivalent to
  $$
    \dot{\phi}^i_t = M^i_t -\frac{1}{\lambda} \int_0^t (\mu^i_s -\gamma \phi^i_s ) ds, \quad \dot{\phi}^i_T=0
  $$
  for some $Q_i$-martingale $M^i$. Put differently, $\dot{\phi}\in\cL^{4}(Q_{i})$ is optimal if and only if it solves the linear forward-backward SDE
  \begin{align}
  d\phi_t  &= \dot{\phi}_t dt, \quad \phi_0=a_i, \label{eq:fwd}\\
  d\dot{\phi}_t &= \frac{\gamma}{\lambda} \left(\phi_t - \frac{\mu^i_t}{\gamma}\right) dt + dM_t, \quad \dot{\phi}_T=0. \label{eq:bwd}
  \end{align}
  Direct computation shows that $(\phi^{i},\dot\phi^{i})$ solves this system: \eqref{eq:fwd} is trivial and for \eqref{eq:bwd} we note that
  \begin{align*}
  d\dot{\phi}^i_t=& \bigg\{\left(\frac{G''(t)}{G(t)}-\frac{G'(t)^2}{G(t)^2}\right)\phi^i_t +\frac{G'(t)}{G(t)}\left(\frac{G'(t)}{G(t)}\phi^i_t+ E^i_t\left[ \int_t^T \frac{G(s)}{G(t)}\frac{\mu^i_s}{\lambda}\,ds\right]\right)\\
  & \quad \;-\frac{G'(t)}{G(t)^2}E^i_t\left[ \int_t^T G(s)\frac{\mu^i_s}{\lambda}\,ds\right]-\frac{\mu^i_t}{\lambda}\bigg\}dt\\
  &+\frac{1}{G(t)}dE^i_t\left[\int_0^T \frac{G(s)\mu^i_s}{\lambda}\,ds\right]= \frac{\gamma}{\lambda}\left(\phi^i_t-\frac{\mu^i_t}{\gamma}\right)dt+dM_t
  \end{align*}
  for the $Q_i$-martingale $M=\int_0^\cdot \frac{1}{G(t)}dE^i_t[\int_0^T \frac{G(s)\mu^i_s}{\lambda}\,ds]$, where $G''(t)=\frac{\gamma}{\lambda}G(t)$ was used. As $G'(T)=0$, the terminal condition $\dot{\phi}^i_T=0$ is also satisfied.
\end{proof}


\section{Equilibrium}\label{se:equilibrium}

The following result establishes the existence and uniqueness of an equilibrium price function $v\in C^{1,2}_{b}([0,T]\times\R^{d})$ and its characterization through a weakly coupled system of linear parabolic equations. Recall that the function $G$ was defined in~\eqref{eq:defG} as $G(t)=\cosh(\sqrt{(\gamma/\lambda)}(T-t))$.

\begin{theorem}
\label{th:equilibriumPDE}
Let $\gamma>0$, $\lambda>0$. The parabolic system
\begin{align}
&\partial_{t}v_i+\frac{1}{2} \tr \sigma_i^2\partial_{xx}v_i+b_i\partial_{x}v_i+\frac{G'(t)}{G(t)}(v_i-v)=0, \quad 1\leq i\leq N, \label{eq:PDEvi}\\
&v:=\frac{1}{N}\sum_{i=1}^Nv_i+\frac{\lambda G'(t)}{N G(t)}a_{0}, \label{eq:PDEv}\\
&v_i(T,\cdot)=f, \quad 1\leq i\leq N  \label{eq:PDEterminalCond}
\end{align} 
has a unique solution $v_{1},\dots,v_{N}\in C^{1,2}_{b}([0,T]\times\R^{d})$, and the function $v$ defined via~\eqref{eq:PDEv} is an equilibrium price function. It is unique in the sense that any equilibrium price function $w\in C^{1,2}([0,T]\times\R^{d})$ with polynomial growth must be equal to~$v$. The equilibrium portfolios are given by
  \begin{equation*}\label{eq:portfolioInEquilibriumPDE}
  \phi^i_t=\frac{G(t)}{G(0)}a_i + \int_0^t\frac{G(t)}{G(s)}E^i_s\left[\int_s^T\frac{G(u)}{G(s)}\frac{\cL^{i}v(u,X_{u})}{\lambda}\,du\right]ds.
  \end{equation*}
\end{theorem}

For comparison, let us first consider how this result simplifies for homogeneous beliefs. Then, all drift and volatility coefficients, and in turn the functions $v_i$ in Theorem~\ref{th:equilibriumPDE}, coincide. Together with the definition of the function $G$, it follows that
$$
v(t,x)=E_{t,x}\left[f(X_T)\right] - \frac{\gamma a_0}{N}(T-t), 
$$
where the expectation is taken under the common probability measure. This is the same equilibrium price that obtains in the frictionless version of the model where trading costs $\lambda$ are zero; cf.\ Proposition~\ref{pr:NoTransactionCostPrice}. Whence, with homogeneous beliefs and holding costs, equilibrium prices do not depend on liquidity; see \cite{sannikov.skrzypacz.16,BouchardFukasawaHerdegenKarbe.18,herdegen.al.19} for similar results. In contrast, the corresponding equilibrium trading speed depends on both the trading cost $\lambda$ and the holding cost $\lambda$ through their ratio. With homogenous beliefs, it is deterministic and the same for all agents $1\leq i \leq N$,
$$
\dot{\phi}^i_t = \sqrt{\frac{\gamma}{\lambda}}\tanh\left(\sqrt{\frac{\gamma}{\lambda}}(T-t)\right)\left(\frac{a_0}{N}-\phi^i_t\right), \quad \phi^i_0=a_i.
$$ 
To wit, the (identical) agents simply gradually adjust their initial allocations until the total supply of the risk asset is split equally among all of them.

Returning to the general case, an immediate consequence of Theorem~\ref{th:equilibriumPDE}  is that the holding costs $\gamma$ and transaction costs $\lambda$ have dual roles. In particular, in the case of zero net supply $a_{0}=0$, the equilibrium price depends only on the ratio $\gamma/\lambda$, so that small transaction costs are equivalent to large holding costs. When $a_{0}>0$, the theorem shows that the equilibrium price is 0-homogeneous in $(\gamma, \lambda,1/a_{0})$. We discuss this in more detail in Section~\ref{se:holdingAsymps} below, after deriving the limiting cases for small costs.

If all agents have equivalent beliefs, one could also consider heterogeneous holding costs $\gamma_i$ as previously studied in models with homogeneous beliefs~\cite{BouchardFukasawaHerdegenKarbe.18,herdegen.al.19} or exogenous price dynamics~\cite{BayraktarMunk.17,casgrain.jaimungal.18}. However, as in those studies, all agents' current positions would then appear as additional state variables and the system would become semilinear. (The extra state variables cancel in the present setting because the function $G(t)$ is the same for all agents.) We thus focus on homogeneous holding costs to highlight the effect of heterogeneity of beliefs.

As a preparation for the proof of Theorem~\ref{th:equilibriumPDE}, we first establish the analytic properties of the parabolic system. Given $\alpha\in(0,1)$, the H\"older space $C^{1+\alpha/2,2+\alpha}([0,T]\times\R^{d})$ consists of the functions $w(t,x)$ such that $w, \partial_{t}w, \partial_{x}w, \partial_{xx}w$ exist, are bounded, and uniformly H\"older continuous with exponents~$\alpha/2$ in $t$ and~$\alpha$ in~$x$.

\begin{proposition}\label{pr:PDEsol}
The system (\ref{eq:PDEvi}--\ref{eq:PDEterminalCond}) has a unique solution $v_{1},\dots,v_{N}\in C^{1,2}_{b}([0,T]\times\R^{d})$.
In fact, $v_{1},\dots,v_{N}\in C^{1+\alpha/2,2+\alpha}([0,T]\times\R^{d})$ for all $\alpha\in(0,1)$  
and uniqueness holds in the larger class of functions $w_{1},\dots,w_{N}\in C^{1,2}([0,T)\times\R^{d})\cap C([0,T]\times\R^{d})$ satisfying $|w_{i}(t,x)|\leq c_{1}\exp(c_{2}|x|^{2})$ for some constants $c_{1},c_{2}\geq0$.
\end{proposition}

\begin{proof}
  The system (\ref{eq:PDEvi}--\ref{eq:PDEterminalCond}) is a weakly coupled, uniformly parabolic linear system; see \cite[Chapter~9]{Friedman.64} for background. Uniqueness in the stated class is a special case of~\cite[Theorem~1]{Besala.63}.
  An existence result for such linear systems is contained, e.g., in \cite[Theorem~3, p.\,256]{Friedman.64}, but this does not yield growth estimates of the type we require here. Our system is also covered by a literature on reaction--diffusion systems. Specifically, \cite[Theorem~2.4]{BechererSchweizer.05} yields that (\ref{eq:PDEvi}--\ref{eq:PDEterminalCond}) has a unique solution $v_{1},\dots,v_{N}\in C^{1,2}([0,T)\times\R^{d})\cap C_{b}([0,T]\times\R^{d})$.
%
%
The main point (which we have not found provided in the literature) is to prove a useful growth estimate on the derivatives, and for that, the key element is to provide a H\"older estimate for~$v_{i}$.

(i) In this step we show that  $v_{i}$ is globally Lipschitz in~$x$, uniformly in~$t$.
Writing $u=(u_{1},\dots,u_{N})$, our system is of the general form
\begin{equation}\label{eq:MNgenSystem}
  \cL^{i} u_{i}(t,x) + h_{i}(t,u(t,x))=0,\quad u_{i}(T,x)=f_{i}(x),\quad 1\leq i\leq N
\end{equation}
satisfying the following conditions, for some constant $c>0$: the function $h=(h_{1},\dots,h_{n})$ is jointly Lipschitz with norm $\Lip(h)\leq c$ (hence $h(t,\cdot)$ is of linear growth, uniformly in $t$); the coefficients of $\cL^{i}$ are bounded and Lipschitz; each $f_{i}$ is bounded and Lipschitz with norm $\Lip(f_{i})$.
According to \cite[Theorem~2.4]{BechererSchweizer.05}, such a system has a (unique) solution $v_{1},\dots,v_{N}\in C^{1,2}\cap C_{b}$. Indeed, define
$$
  F_{i}(u)(t,x) := E^{i}\left[f_{i}(X^{t,x}_{T}) + \int_{t}^{T} h_{i}(s,u(s,X^{t,x}_{s})) ds \right], \quad 1\leq i \leq N.
$$
It is shown in the proofs of \cite[Theorems~2.3 and~2.4]{BechererSchweizer.05} that $F=(F_{1},\dots, F_{N})$ is a contraction on $(C_{b})^{N}=C_{b}\times\dots\times C_{b}$ for a complete norm which is equivalent to the uniform norm. More precisely, this holds after suitably truncating $h$ (so that the truncation does not affect the bounded solution). It is shown that if we start at any $u\in (C_{b})^{N}$ and iterate $F$, the sequence $u^{n}=(F\circ \dots\circ F)(u)$ will converge uniformly to a solution $(v_{1},\dots,v_{N})\in (C^{1,2}\cap C_{b})^{N}$ of~\eqref{eq:MNgenSystem}. We may, in particular, pick $u\in (C_{b})^{N}$ such that $\sup_{0\leq s\leq T}\Lip (u_{i}(s,\cdot))<\infty$ for all $1\leq i\leq N$ as our starting point for the iteration.

By a standard estimate on SDEs (e.g., \cite[Theorem~2.4\,(i),\, p.\,8]{Touzi.13}), 
$$
  E^{i}|X^{t,x}_{s}-X^{t,y}_{s}| \leq K|x-y|, \quad 0\leq s\leq T
$$
for a constant $K$ depending only on the Lipschitz constants of the coefficients of~$\cL^{i}$ and~$T$. Fix a small time interval $[t,T]$ of length $\tau=T-t>0$ and let
$$
  L_{u}=L_{u}^{(t)}= \max_{1\leq i\leq N} \sup_{t\leq r\leq T}\Lip(u_{i}(r,\cdot)).
$$
Then for $t\leq s\leq T$ we have that
\begin{align*}
  &|F_{i}(u)(s,x) - F_{i}(u)(s,y)| \\
  & \leq E^{i}\left[\Lip(f_{i})|X^{s,x}_{T}-X^{s,y}_{T}| + \int_{s}^{T} c\, \textstyle{\max_{i}}\Lip(u_{i}(r,\cdot)) |X^{s,x}_{r}-X^{s,y}_{r}| dr \right] \\
  & \leq  [K\Lip(f_{i}) + \tau c K L_{u}]\, |x-y|.
\end{align*}
This holds for all $i$. Choose $\tau$ such that $\eps:=\tau c K<1$ and set $L_{f}=\max_{1\leq i\leq N}K\Lip(f_{i})$, then
$$
  \Lip(F(u)(s,\cdot)) \leq L_{f} + \eps L_{u}, \quad t\leq s\leq T,
$$
the notation of course meaning that each component $F_{i}(u)$ satisfies this property. Iterating yields that $u^{n}=(F\circ \dots\circ F)(u)$ satisfies the geometric estimate
$$
  \Lip(u^{n}(s,\cdot)) \leq L_{f} \sum_{k=0}^{n-1} \eps^{k} + \eps^{n} L_{u}
$$
and hence the uniform limit $(v_{1},\dots,v_{N})=\lim u^{n}$ satisfies $\Lip(v_{i}(s,\cdot))\leq L_{f} (1-\eps)^{-1}$ for $t\leq s\leq T$.

Note that the size $\tau$ of the interval in the above argument does not depend on $\Lip(f)$. Hence we can repeat the argument on the interval $[T-2\tau, T-\tau]$, replacing the terminal condition $f_{i}$ by $\tilde f_{i} :=v_{i}(T-\tau,\cdot)$. Continuing finitely many times, we conclude that $\sup_{0\leq s\leq T}\Lip(v_{i}(s,\cdot))<\infty$.

(ii) Next, we show that $v_{i}$ is globally $1/2$-H\"older in~$t$, uniformly in $x$. A simple SDE estimate shows that 
$$
  E^{i}|X^{t',x}_{s}-X^{t,x}_{s}| \leq K|t'-t|^{1/2}, \quad 0\leq t\leq t' \leq s \leq T
$$
where $K$ now depends on the Lipschitz constants and uniform bounds for  $b_{i}$ and $\sigma_{i}$ as well as~$T$. (To see this one may, e.g., go through the proof of \cite[Theorem~2.4\,(ii),\, p.\,8]{Touzi.13} and use the uniform bounds in the estimate below Equation~(2.5) of that reference to avoid a dependence on $x$ in the final estimate for $E^{i}|X^{t,x}_{s}-X^{t,y}_{s}|$.)

As mentioned above, the relevant function $h$ in~\eqref{eq:MNgenSystem} is truncated in~$u$, so that $\|h\|_{\infty}<\infty$. This yields the (crude but simple) estimate
\begin{equation}\label{eq:MNcrudeTimeEstimate}
  \int_{t}^{t'} |h(s,X^{t,x}_{s},u(s,X^{t,x}_{s}))| ds \leq \|h\|_{\infty}|t'-t|  \leq c'|t'-t|^{1/2}
\end{equation}
for some constant $c'$, because $|t'-t|\leq T$. If $u\in (C_{b})^{N}$ is Lipschitz in $x$ with constant $L_{u}$ uniformly in $t$, we then have, similarly as in~(i) but using also~\eqref{eq:MNcrudeTimeEstimate},
\begin{align*}
  &|F_{i}(u)(t',x) - F_{i}(u)(t,x)| \\
  & \leq E^{i}\left[\Lip(f_{i})|X^{t',x}_{T}-X^{t,x}_{T}| + c'|t'-t|^{1/2} +  \int_{t'}^{T} c L_{u}|X^{t',x}_{s}-X^{t,x}_{s}| ds \right] \\
  & \leq  [K\Lip(f_{i}) + c' + Tc L_{u}K]\, |t'-t|^{1/2}.
\end{align*}
Again, we iterate the mapping $F$ to generate $u^{n}=(F\circ \dots\circ F)(u)$. By~(i) we have that $\sup_{n} L_{u^{n}}<\infty$. Hence, the above shows that $|u^{n}_{i}(t',x)-u^{n}_{i}(t,x)|\leq c'' |t'-t|^{1/2}$ for a uniform constant $c''$, and then the same holds for the limit $(v_{1},\dots,v_{N})=\lim u^{n}$.

(iii) We have shown above that $v$ is globally Lipschitz in $x$ and 1/2-H\"older in $t$; in particular $v_{j}\in C^{\alpha/2,\alpha}$ for all $\alpha\in (0,1)$. (See \cite[p.\,117]{Krylov.96} for a detailed definition of the H\"older spaces.) For fixed~$i$, we can see $v_{i}$ as the solution of a \emph{scalar} PDE which contains $(v_{j})_{j}$ as coefficients: $\varphi=v_{i}$ is the solution of
$$
 \tilde{\cL}\varphi(t,x) + g(t,x) =0, \quad \varphi(T,\cdot)=f
$$
on $[0,T]\times \R^{d}$ with terminal value $f\in C^{2+\alpha}$, parabolic operator $\tilde\cL u := \cL^{i}u-u$ and  inhomogeneous term $g\in C^{\alpha/2,\alpha}$ defined by
$$
  g = v_{i} +\frac{G'(t)}{G(t)}\left(v_i-\frac{1}{N}\sum_{i=1}^Nv_i+\frac{\lambda G'(t)}{N G(t)}a_{0}\right)
$$
using the fixed functions $(v_{j})_{1\leq j\leq N}$. We can now apply a suitable version of the Schauder estimates to conclude that $v_{i}=\varphi\in C^{1+\alpha/2,2+\alpha}([0,T]\times\R^{d})$; cf.\ \cite[Theorem~9.2.3, p.\,140]{Krylov.96}.
\end{proof}

\begin{remark}\label{rk:PDEsolSmooth}
  Suppose that $b_{i},\sigma_{i},f\in C^{\infty}_{b}$ for $1\leq i\leq N$. Then we also have $v_{i}\in C^{\infty}_{b}$. 
\end{remark} 

\begin{proof}
  If $b_{i},\sigma_{i}\in C^{\infty}([0,T)\times\R^{d})$, interior regularity for parabolic systems as stated in \cite[Theorem~11, p.\,265]{Friedman.64} immediately yields that the solution from Proposition~\ref{pr:PDEsol} is in $C^{\infty}([0,T)\times\R^{d})$. We need to show that the partial derivatives of any order are bounded.

  Fix $1\leq i\leq N$ and $1\leq k\leq d$ and consider the function $\varphi=\partial_{x_{k}} v_{i}$. We can differentiate the system~\eqref{eq:PDEvi} with respect to $x_{k}$ and rearrange the terms to find that $\varphi$ is the solution of a scalar parabolic equation
$$
  \cL\varphi(t,x) + g(t,x) =0, \quad \varphi(T,\cdot)=\partial_{x_{k}}f
$$
on $[0,T]\times \R^{d}$ with terminal value $\partial_{x_{k}}f\in C^{\infty}_{b}\subseteq C^{2+\alpha}$. Here the inhomogeneity $g$ incorporates all other terms resulting from the differentiated equation: it is a linear combination, with coefficients in $C^{\infty}([0,T)\times\R^{d})$, of the functions $v_{j}$, $1\leq j\leq N$ as well as their first and second-order spatial derivatives. As $v_{j}\in C^{1+\alpha/2, 2+\alpha}$ by Proposition~\ref{pr:PDEsol}, we see in particular that $g\in C^{\alpha/2, \alpha}$. Thus, we can conclude from \cite[Theorem~9.2.3, p.\,140]{Krylov.96} that $\partial_{x_{k}}v_{i}=\varphi\in C^{1+\alpha/2,2+\alpha}([0,T]\times\R^{d})$. In particular, the third-order spatial derivatives of $v_{i}$ are bounded and uniformly H\"older continuous. Moreover, by the parabolic form of the above equation, the same follows for $\partial_{t}\partial_{x_{k}}v_{i}=\partial_{t}\varphi$.

This argument can be iterated to the higher-order derivatives.
\end{proof} 

\begin{proof}[Proof of Theorem~\ref{th:equilibriumPDE}.]
  The formula for the equilibrium portfolios is a direct consequence of Proposition~\ref{pr:NoTransactionCostPrice}, so we focus on the price.

  (i) Let $v_{1},\dots,v_{N}\in C^{1,2}_{b}$ be the solution from Proposition~\ref{pr:PDEsol} and define $v$ by~\eqref{eq:PDEv}; we show that $v$ is an equilibrium price function. It\^o's formula shows that $S_{t}=v(t,X_{t})$ is a price process as defined in~\eqref{eq:SunderQi}; the coefficients $\mu_{i}$ and $\nu_{i}$ are even bounded. In view of~\eqref{eq:PDEv}, the function $w_{i}(t,x):=G(t)v_{i}(t,x)$ satisfies
  $$
    \cL^{i}w_{i} = G\cL^{i}v_{i} + G'v_{i} =G'v
  $$
  and $w_{i}(T,x)=G(T)v_{i}(T,x)=f(x)$. Thus, It\^o's formula and the boundedness of~$\partial_{x}v_{i}$  imply that under~$Q_{i}$ we have the Feynman--Kac representation 
  $$
    w_i(t,x)= E^i_{t,x}[f(X_T)]-\int_t^T G'(u)E^{i}_{t,x}[v(u,X_u)] du.
  $$
  As a result,
  \begin{equation}\label{eq:viFeynman}
    v_i(t,x)=\frac{E^i_{t,x}[f(X_T)]}{G(t)}-\int_t^T\frac{G'(u)}{G(t)}E^{i}_{t,x}[v(u,X_u)] du   
  \end{equation} 
  for all $(t,x)\in [0,T]\times\R^d$.
  Lemma~\ref{le:optPortfolio} shows that given the price $S_{t}=v(t,X_{t})$, the portfolio
  \begin{equation}\label{eq:optPortfolioS}
    \phi^i_t=\int_0^t\frac{G(t)}{G(s)}E^i_{s,X_s}\Big[\int_s^T\frac{G(u)}{G(s)}\frac{1}{\lambda}\cL^{i}v(u,X_u)\,du\Big] ds+\frac{G(t)a_i}{G(0)}
  \end{equation}
  is optimal for agent~$i$. It remains to prove that these portfolios clear the market. Recalling that $G(T)=1$, taking expectations in the integration-by-parts formula
  $\int_{s}^{T} G(u) dS_{u}=G(T)S_{T} - G(s)S_{s} - \int_{s}^{T} G'(u)S_{u} du$ and applying~\eqref{eq:viFeynman}   
  yield 
  \begin{align}
  &E^i_{s,x}\left[\int_s^T\frac{G(u)}{G(s)}\frac{1}{\lambda}\cL^{i}v(u,X_u)\,du\right] \notag\\
  &=\frac{1}{\lambda G(s)}\left( E^i_{s,x}[f(X_T)]-G(s)v(s,x)-\int_s^T G'(u)E^i_{s,x}[v(u,X_u)] du\right) \notag\\
  &=\frac{1}{\lambda} [v_{i}(s,x)-v(s,x)] \label{eq:intByPartPDEproof}.
  \end{align}
  In view of~\eqref{eq:PDEv}, we deduce that
  $$
    \sum_{i=1}^{N }E^i_{s,x}\left[\int_s^T\frac{G(u)}{G(s)}\frac{1}{\lambda}\cL^{i}v(u,X_u)\,du\right] = -\frac{G'(s)}{G(s)}a_{0}.
  $$
  Using this in~\eqref{eq:optPortfolioS} and integrating $-\frac{G'(s)}{G^{2}(s)}=\partial_{s}G(s)^{-1}$, we conclude that
  $$
    \sum_{i=1}^{N} \phi^{i}_{t} = -a_{0}\int_0^t\frac{G(t)}{G(s)} \frac{G'(s)}{G(s)} \,ds+\frac{G(t)a_0}{G(0)}
    =a_{0}
  $$  
  as desired.
  
  (ii) Let $S_{t}=w(t,X_{t})$ be an equilibrium price process for some function $w\in C^{1,2}([0,T]\times\R^{d})$
  of polynomial growth (or, more generally, $w\in C^{1,2}([0,T)\times\R^{d})$ of polynomial growth and locally H\"older continuous on $[0,T]\times\R^{d}$). We have $w(T,\cdot)=f$ by Remark~\ref{rk:support}. Recall that the coefficients $\mu^{i}_{t}=\cL^{i}w(t,X_{t})$ and $\nu^{i}_{t}=\partial_{x}w(t,X_{t})^{\top}\sigma^{i}_{t}$ of~\eqref{eq:SunderQi} are in $\cL^{4}(Q_{i})$ as part of our definition of a price process.
  We define $w_{i}$ by the Feynman--Kac formula~\eqref{eq:viFeynman} with~$w$ instead of~$v$. In view of the assumptions on $b_{i}$ and $\sigma_{i}$, the function $w_{i}$ has polynomial growth like~$w$. Moreover, by a careful application of standard PDE results, $w_{i}\in C^{1,2}$ and $w_{i}$ is a solution of the associated linear PDE~\eqref{eq:PDEvi}. Specifically, we can use an approximation with bounded domains as detailed in \cite[Theorem~1, Condition~(A3'), Lemma~2 and the comments above it]{HeathSchweizer.00} under the stated conditions on~$w$.

  It remains to show~\eqref{eq:PDEv}. As a consequence of Lemma~\ref{le:optPortfolio}, the agents' equilibrium portfolios~$\phi^{i}$ are given by~\eqref{eq:optPortfolioS}. Because these portfolios clear the market, $\sum_{i}\phi^{i}_{s}=a_{0}$ and thus $\partial_{s} \sum_{i} \frac{\phi^{i}_{s}}{G(s)}=a_{0}\partial_{s}G(s)^{-1}$. Reversing the integration by parts~\eqref{eq:intByPartPDEproof}, we conclude that
  \begin{align*}
  \sum_{i=1}^{N} \frac{1}{G(s)\lambda}[w_{i}(s,x)-w(s,x)]
  &=\sum_{i=1}^{N} \frac{1}{G(s)}E^i_{s,x}\left[\int_s^T\frac{G(u)}{G(s)} \frac{1}{\lambda}\cL^{i}w(u,X_u)\,du\right]\\
  &= \partial_{s}\sum_{i=1}^{N} \frac{\phi^{i}_{s}}{G(s)}
  =a_{0}\partial_{s}G(s)^{-1} 
  = -a_{0}\frac{G'(s)}{G(s)^{2}}
  \end{align*} 
  which is equivalent to~\eqref{eq:PDEv}. We have thus established that $w_{1},\dots,w_{N}\in C^{1,2}$ are a solution of~\eqref{eq:PDEvi} with polynomial growth. The claim now follows by the uniqueness of the solution as stated in Proposition~\ref{pr:PDEsol}.
\end{proof} 

\begin{remark}\label{rk:nonMarkovCase}
  The restriction to Markovian equilibria in Theorem~\ref{th:equilibriumPDE} (meaning that the price is a function of $t$ and $x$) is related to our choice of proof through PDE arguments rather than fundamental. For instance, if the volatility $\sigma_{i}$ is the same for all agents, similar arguments could be carried out using backward SDEs (e.g., \cite{ElKarouiPengQuenez.97}). In that framework, one would obtain uniqueness within a class of possibly non-Markovian equilibria and one could also cover beliefs where~\eqref{eq:Xdynamics} is replaced by coefficients that may depend on the path of~$X$.
\end{remark} 

\section{Asymptotics for Small Transaction Costs}\label{se:transactionAsymp}

In this section we provide intuition for the equilibrium price from Theorem~\ref{th:equilibriumPDE} by describing its asymptotics for small transaction costs $\lambda\to 0$.
For later comparison, we first consider the model without transaction costs; i.e., $\lambda=0$. In this case we drop the requirement of absolute continuity in the definition of the admissible portfolios and we also do not enforce the initial holdings~$a_{i}$ (in any event, agents can instantaneously adjust their position after $t=0$ without incurring costs). The following result, which is a special case of~\cite[Theorem~2.1 and Remark~3.5]{NutzScheinkman.17}, shows that the corresponding equilibrium corresponds to the price of a representative agent with a view $\bar Q$ defined by the averaged drift and volatility coefficients.

\begin{proposition}\label{pr:NoTransactionCostPrice}
  Let $\lambda=0$ and $\gamma>0$. There exists a unique equilibrium price function $v^{0}\in C_b^{1,2}$, given by
  $$
    v^{0}(t,x)=\bar{E}_{t,x}[f(X_T)]-\frac{(T-t)\gamma a_{0}}{N}
  $$
  where $\bar{E}[\cdot]$ is the expectation for the probability $\bar{Q}$ under which
  $$
 dX_t=\bar{b}(t,X_t) dt+\bar{\sigma}(t,X_t) dW_t,\quad \textstyle{\bar{b}=\frac{1}{N}\sum_{i=1}^Nb_i, \quad \bar{\sigma}^{2}=\frac{1}{N}\sum_{i=1}^N\sigma_i^2}
  $$
 for a Brownian motion $W$. Equivalently, $v^{0}$ is the unique bounded classical solution of 
  \begin{equation}\label{eq:PDEv0}
    \partial_{t}v+\frac{1}{2} \tr \bar\sigma^2\partial_{xx}v +\bar{b}\partial_{x}v - \frac{\gamma a_{0}}{N}  =0, \quad v(T,\cdot)=f.
  \end{equation}
  In equilibrium, the $dt\times Q_{i}$-a.e.\ unique optimal portfolio for agent $i$ is
  \begin{equation}\label{eq:phi0}
    \phi^{i,0}_t=\frac{\cL^{i}v^0(t,X_t)}{\gamma}.
  \end{equation}
\end{proposition}

In the remainder of this section we denote the equilibrium price from Theorem~\ref{th:equilibriumPDE} by $v^\lambda$ to emphasize the dependence on  $\lambda$. Our goal is to compute its leading-order deviation $\lambda^{-1/2}(v^\lambda-v^0)$ from the frictionless equilibrium price $v^0$ of Proposition~\ref{pr:NoTransactionCostPrice}. For simplicity, we focus on the case of a one-dimensional state variable ($d=1$) with smooth drift and diffusion coefficients and terminal condition: $b_{i},\sigma_{i},f\in C^{\infty}_{b}$ for $1\leq i\leq N$, and hence $v,v_{i}\in C^{\infty}_{b}$ on the strength of Remark~\ref{rk:PDEsolSmooth}.\footnote{For the arguments below, $C^{9}_{b}$ is sufficient. Some weakening of these regularity assumptions is certainly possible, for the ease of exposition we have retained the given smoothness assumptions.}

\begin{theorem}
\label{theorem:TransactionCostExpansion}
For fixed holding costs $\gamma>0$ and small transaction costs $\lambda \to 0$, the equilibrium price function $v^\lambda$ from Theorem~\ref{th:equilibriumPDE} has the expansion
\begin{equation}\label{eq:singular}
v^\lambda(t,x)=v^0(t,x)+\sqrt{\lambda} v^{*}(t,x)+o(\sqrt{\lambda}) \quad \mbox{locally uniformly on $[0,T]\times\mathbb{R}$}.
\end{equation}
Here, $v^0$ is the frictionless equilibrium price from Proposition \ref{pr:NoTransactionCostPrice} and 
\begin{align}\label{eq:ve}
v^{*}(t,x)&=\frac{\sqrt{\gamma}}{N}\sum_{i=1}^N \bar{E}_{t,x}\left[\int_t^T \mathcal{L}^i \hat\phi^{i,0}(s,X_{s})ds\right]
\end{align}
where $\hat\phi^{i,0}(s,x)=\mathcal{L}^i v^0(s,x)/\gamma$ is the feedback function determining agent~$i$'s frictionless optimal portfolio~\eqref{eq:phi0} and the expectation is taken under the probability measure $\bar{Q}$ of the frictionless representative agent for which
$$
 dX_t=\bar{b}(t,X_t) dt+\bar{\sigma}(t,X_t) dW_t,\quad \textstyle{\bar{b}=\frac{1}{N}\sum_{i=1}^Nb_i, \quad \bar{\sigma}^{2}=\frac{1}{N}\sum_{i=1}^N\sigma_i^2}.
$$
\end{theorem}

The singular perturbation expansion~\eqref{eq:singular} shows that the leading-order deviation of the frictional equilibrium price $v^\lambda$ from its frictionless counterpart $v^0$ scales with the square root $\sqrt{\lambda}$ of the trading cost, as in the risk-sharing equilibrium of~\cite{herdegen.al.19}. 
With the heterogeneous beliefs considered in the present study, the constant of proportionality~\eqref{eq:ve} is determined by the integrated drift rates $\int_t^T \mathcal{L}^i \hat\phi^{i,0}(s,X_{s})ds$ of the agents' frictionless equilibrium portfolios, averaged with respect to agents and states. Thus, equilibrium prices increase relative to their frictionless counterparts if agents on average expect to increase their positions in the future, and vice versa.\footnote{
Note that while the actual future portfolio changes add up to zero by market clearing, this is not necessarily true for the changes as anticipated by the heterogeneous agents under their subjective probability measures, $\mathcal{L}^i \hat\phi^{i,0}$. In the formula for $v^{*}$, these anticipated changes are averaged across all states under the probability measure $\bar{Q}$ corresponding to the frictionless representative agent.}
The interpretation is that in illiquid markets, agents take into account their future trading needs to save cumulative transaction costs over the whole time interval. Accordingly, expectations of future purchases already lead to increased positions earlier on, and vice versa. To clear the market, equilibrium prices increase or decrease according to the excess demand or supply created by the aggregated adjustments of all agents.

\subsection{Proof of Theorem~\ref{theorem:TransactionCostExpansion}}

The first step towards the proof of Theorem~\ref{theorem:TransactionCostExpansion} is to show that the functions $v_i^\lambda$ from Theorem~\ref{th:equilibriumPDE} are not just bounded for each $\lambda$, but that this bound is in fact uniform for $\lambda \in (0,\infty)$. In view of the PDEs~\eqref{eq:PDEvi} from Theorem~\ref{th:equilibriumPDE} and as $\frac{\lambda G^\prime(t)^2}{NG(t)^2}a_0$ is uniformly bounded for all $\lambda >0$ and $t\in [0,T]$ by the definition of~$G$, this is a special case of the following more general result that will also allow us to derive estimates for small holding costs in the subsequent section.

\begin{lemma}
\label{lemma:UniformBound}
For $i=1,\dots,N$ and an arbitrary parameter $\varepsilon \in \mathscr{E}$, consider 
functions $\alpha_i$, $\beta_i$, $a_i^\varepsilon$,  $b_i^\varepsilon$, $h_i \in C^{\infty}_{b}([0,T] \times \mathbb{R})$ and write
$$\mathcal{L}^i=\partial_t+\frac{1}{2}\beta_i^2\partial_{xx}+\alpha_i\partial_x.$$
Suppose that $a_i^\varepsilon, b_i^\varepsilon$ are bounded uniformly in $\varepsilon\in \mathscr{E}$ and let $u_i=u_i(\eps,\lambda,\gamma)$, $i=1,\ldots,N$ denote the unique classical bounded solution of the 
system
$$\mathcal{L}^iu_i+\Big(\frac{G^\prime}{G}+a_i^\varepsilon\Big)u_i-\frac{G^\prime}{G}\frac{1}{N}
\sum_{j=1}^Nu_j+b_i^\varepsilon=0, \quad u_i(T,\cdot)=h_i, \quad i=1,\ldots,N.
$$ 
Then, $|u_i(t,x)|\leq M$ for a constant $M>0$ independent of $\varepsilon, \lambda, \gamma \in (0,\infty)$ and $(t,x)\in[0,T]\times\mathbb{R}$.

\end{lemma}

\begin{proof}
Existence and uniqueness of the $u_i$ is a special case of~\cite[Theorem~2.4]{BechererSchweizer.05}. Because these functions are bounded, the Feynman--Kac formula as in \cite[Theorem 5.7.6]{KaratzasShreve.91} as well as $G'/G=(\log G)'$ and $G(T)=1$ give
\begin{align*}
e^{\int_0^ta_i^\varepsilon d\tau}u_i(t,x)&=E^i_{t,x}\Bigg[\int_t^T-\frac{G'(s)}{G(t)}\frac{1}{N}\sum_{j=1}^Ne^{\int_0^sa_i^\varepsilon d\tau}u_j(s,X_s)ds\\
&\qquad\qquad +\int_t^Te^{\int_0^sa_i^\varepsilon d\tau}\frac{G(s)}{G(t)}b_i^\varepsilon ds+e^{\int_0^Ta_i^\varepsilon d\tau}\frac{1}{G(t)}h_i(X_T)\Bigg]
\end{align*}
where the expectation is taken under the measure for which the state variable has dynamics $dX_t=\alpha_i(t,X_t)dt+\beta_i(t,X_t)dW^i_t$. Choose a uniform bound $M$ for $\big|e^{\int_0^sa_i^\varepsilon d\tau}b_i^\varepsilon\big|$ and $\big|e^{\int_0^Ta_i^\varepsilon d\tau}h_i\big|$, and define 
\begin{align*}
K(t,\lambda, \gamma)=\max\left\{|e^{\int_0^ta_i^\varepsilon(\tau,x_\tau) d\tau}u_i(t,x_t)|\right\}<\infty,
\end{align*}
where the maximum is taken over  $i \in \{1,\ldots,N\}$, $\varepsilon \in \mathscr{E}$, and $(x_\tau)_{\tau \in [0,t]} \in C_0([0,t],\mathbb{R})$. With this notation,
\begin{align*}
|e^{\int_0^ta_i^\rho d\tau}u_i(t,x)|\leq \int_t^T -\frac{G^\prime(s)}{G(t)}K(s,\lambda, \gamma)ds+M\int_t^T\frac{G(s)}{G(t)}ds+\frac{M}{G(t)},
\end{align*}
which in turn leads to
$$G(t)K(t,\lambda, \gamma)\leq \int_t^T\Big(-\frac{G^\prime(s)}{G(s)}\Big)G(s)K(s,\lambda, \gamma)ds+M\int_t^TG(s)ds+M.$$
We may read this as an inequality of the form $u(t)\leq \int_t^T B(s)u(s)ds + A(t)$ for $u(t)=G(t)K(t,\lambda, \gamma)$.
Using $G'/G=(\log G)'$ and that $G$ is decreasing, $\int_t^T G(r)dr\leq G(t)T$ and Gr\"onwall's lemma yield
\begin{align}
G(t)K(t,\lambda,\gamma) \leq M(G(t)T+1)-MG(t)\int_t^T\Big(\int_s^TG(r)dr+1\Big)\frac{G^\prime(s)}{G(s)^2}ds.\label{ineq:GKlambda}
\end{align}
Observe that $G$ satisfies
$G=\frac{\lambda}{\gamma}G''$ and $G'(T)=0$ and $\frac{\lambda}{\gamma}\frac{(G')^2}{G^2} \leq 1$, so that
\begin{align*}
-\int_t^T\Big(\int_s^TG(r)dr\Big)\frac{G^\prime(s)}{G(s)^2}ds&=\int_t^T\frac{\lambda}{\gamma}\frac{G^\prime(s)^2}{G(s)^2}ds\leq T-t\leq T.
\end{align*}
Together with
\begin{align*}
-\int_t^T\frac{G^\prime(s)}{G(s)^2}ds=1-\frac{1}{G(t)}\leq 1,
\end{align*}
it follows from \eqref{ineq:GKlambda} and $G(t) \geq 1$ that $K(t,\lambda,\gamma)\leq 2M(T+1)$. As $a^\varepsilon$ is uniformly bounded in $\varepsilon, t,x$, the function $u_i$ is therefore uniformly bounded in $\varepsilon,\gamma,\lambda,t,x$ by the definition of $K(t,\lambda,\gamma)$.
\end{proof}

\begin{corollary}
\label{lemma:Boundforv}
There exists $M>0$ such that $|v^\lambda_i(t,x)|\leq M$ for all $\lambda>0$ and $(t,x)\in [0,T]\times\mathbb{R}$.
\end{corollary}

Next, we establish an analogous uniform bound for the derivatives of the functions $v^\lambda_i$ and $v^\lambda$ from Theorem~\ref{theorem:TransactionCostExpansion}.

\begin{lemma}
\label{lemma:BoundforDerivativesofv}
Fix $k\geq0$. There exists $M>0$ such that 
$$
|\partial_x^kv_i^\lambda(t,x)|, \, |\partial_x^kv^\lambda(t,x)|, \, |\partial_t\partial_x^kv^\lambda(t,x)|\leq M
$$
for all $\lambda>0$ and $(t,x)\in [0,T]\times\mathbb{R}$.
\end{lemma}

\begin{proof}
By Theorem \ref{th:equilibriumPDE} and Remark \ref{rk:PDEsolSmooth}, the $x$-derivatives of the functions $v_i^\lambda$, $i=1,\ldots,N$ from Theorem~\ref{th:equilibriumPDE} satisfy the following PDEs obtained by differentiating~\eqref{eq:PDEvi} with respect to the spatial variable:
\begin{align}
\partial_t\partial_xv^\lambda_i &+\frac{1}{2}\sigma_i^2\partial_{xx}\partial_xv^\lambda_i+(b_i+\sigma_i\partial_x\sigma_i)\partial_x\partial_xv^\lambda_i \label{eq:para}\\
&+\Big(\partial_xb_i+\frac{G^\prime}{G}\Big)\partial_xv_i^\lambda-\frac{G^\prime}{G}\frac{1}{N}\sum_{j=1}^N\partial_xv^\lambda_j=0, \quad \partial_xv^\lambda_i(T,\cdot)=f^\prime. \notag
\end{align}
Lemma~\ref{lemma:UniformBound} therefore yields the desired uniform bound for $|\partial_xv_i^\lambda(t,x)|$, and in turn also for $\partial_xv^\lambda(t,x)=\frac{1}{N}\sum_{i=1}^N\partial_xv_i^\lambda(t,x)$.   The corresponding bounds for the higher-order $x$-derivatives follow by iterating this argument. 
Finally, the uniform bound for the time derivative of $\partial_x^k v^\lambda$ is then direct consequences of the parabolic form of the PDEs~\eqref{eq:PDEvi}, \eqref{eq:para}, etc., and their sums.
\end{proof}

\begin{lemma}
\label{lemma:UniformConvergence}
For $\lambda>0$, consider $\alpha$, $\beta$, $a^\lambda$, $h$ of class $C^{\infty}_{b}$ and write 
$$\mathcal{L}=\partial_t+\frac{1}{2}\beta^2\partial_{xx}+\alpha\partial_x.$$
Suppose that $w^\lambda \in C^{\infty}_b$
satisfies $w^\lambda(T,\cdot)=h$ and $\partial_tw^\lambda$, $\partial_xw^\lambda$, $\partial_{xx}w^\lambda$, $a^\lambda$ are bounded uniformly in $\lambda$. Then, the unique bounded classical solution $u^\lambda$ of 
$$\mathcal{L}u^\lambda+a^\lambda(t,x)+\frac{G^\prime(t)}{G(t)}(u^\lambda-w^\lambda)=0, \quad u^\lambda(T,\cdot)=h
$$
satisfies 
$$
\frac{|u^\lambda(t,x)-w^\lambda(t,x)|}{\sqrt{\lambda}}\leq M
$$
for some $M>0$ independent of $\lambda>0$ and $(t,x)\in [0,T]\times\mathbb{R}$.
\end{lemma}

\begin{proof}
Following the same steps as in the derivation of~\eqref{eq:intByPartPDEproof} yields that
\begin{align}
\label{equ:differenceofuw}
&u^\lambda(t,x)=w^\lambda(t,x)+E_{t,x}\left[\int_t^T\frac{G(u)}{G(t)}\Big(a^\lambda+\mathcal{L}w^\lambda\Big)(u,X_u)du\right]
\end{align}
where the expectation is taken under the measure for which the state variable has dynamics $dX_t=\alpha(t,X_t)dt+\beta(t,X_t)dW_t$.
With a uniform bound $M$ for 
$a^\lambda+\mathcal{L}w^\lambda$,
the desired bound is 
\begin{align*}
\frac{|u^\lambda(t,x)-w^\lambda(t,x)|}{\sqrt{\lambda}}\leq \frac{M}{\sqrt{\lambda}}\int_t^T\frac{G(u)}{G(t)}du = -\frac{M\sqrt{\lambda}}{\gamma} \frac{G'(t)}{G(t)} \leq \frac{M}{\sqrt{\gamma}}
\end{align*}
where we have once again used $G(u)=\frac{\lambda}{\gamma}G''(u)$ and $G'(T)=1$ in the second step and the definition of $G$ for the last inequality.
\end{proof}


\begin{corollary}
\label{co:viConvergetov}
Fix $k\geq0$. There exists $M>0$ such that
$$\frac{|\partial_x^kv_i^\lambda(t,x)-\partial_x^kv^\lambda(t,x)|}{\sqrt{\lambda}}\leq M$$ 
for all $\lambda>0$ and $(t,x)\in [0,T]\times\mathbb{R}$ and  $i\in \{1,\dots,N\}$.
\end{corollary}

\begin{proof}
In view of the PDEs~\eqref{eq:PDEvi} from Theorem~\ref{th:equilibriumPDE} and the uniform bounds from Lemma \ref{lemma:BoundforDerivativesofv}, Lemma~\ref{lemma:UniformConvergence} yields that $\lambda^{-1/2}|v_i^\lambda-v^\lambda|\leq M$ for some constant $M$. This proves the assertion for $k=0$. The analogous bounds for the derivatives follow by applying the same argument to the corresponding PDEs obtained by differentiating~\eqref{eq:PDEvi} as in the proof of Lemma~\ref{lemma:BoundforDerivativesofv}.
\end{proof}

We can now estimate the difference between $v^{\lambda}$ and the frictionless equilibrium price $v^{0}$ of Proposition~\ref{pr:NoTransactionCostPrice}.

\begin{proposition}
\label{proposition:ZeroOrderConvergenceforSmallTransactionCost}
Fix $k\geq0$. There exists $M>0$ such that 
$$
\frac{|\partial_x^kv^\lambda(t,x)-\partial_x^kv^0(t,x)|}{\sqrt{\lambda}},\,  \frac{|\partial_t\partial_x^kv^\lambda(t,x)-\partial_t\partial_x^kv^0(t,x)|}{\sqrt{\lambda}} \leq M
$$
for all $\lambda>0$ and $(t,x)\in [0,T]\times\mathbb{R}$. 
\end{proposition}

\begin{proof}
Using~\eqref{eq:PDEv} and then~\eqref{eq:PDEvi} for $v$, subtracting the PDE~\eqref{eq:PDEv0} for~$v^0$, and using once again $G(u)=\frac{\lambda}{\gamma}G''(u)$, 
we obtain 
\begin{align}
\label{equ:Differenceofvv0}
&\partial_t(v^\lambda-v^0)+\frac{1}{2}\bar{\sigma}^2\partial_{xx}(v^\lambda-v^0)+\bar{b}\partial_x(v^\lambda-v^0)\\
&\qquad+\frac{1}{N}\sum_{i=1}^N\frac{1}{2}\sigma_i^2(\partial_{xx}v_i^\lambda-\partial_{xx}v^\lambda)+\frac{1}{N}\sum_{i=1}^Nb_i(\partial_xv_i^\lambda-\partial_xv^\lambda)=0\notag
\end{align} 
with $(v^\lambda-v^0)(T,\cdot)=0$. Here $\bar b, \bar\sigma$ are as defined in Proposition~\ref{pr:NoTransactionCostPrice}. The desired uniform bound for $\lambda^{-1/2}|v^\lambda-v^0|$ is now a consequence of the Feynman--Kac formula and Corollary~\ref{co:viConvergetov}. The analogous result for $\lambda^{-1/2}|\partial_x^kv^\lambda-\partial_x^kv^0|$ follows from the same argument because Remark~\ref{rk:PDEsolSmooth} shows that these derivatives satisfy similar PDEs obtained by differentiating~\eqref{equ:Differenceofvv0}. The corresponding bounds for the time derivatives in turn are a consequence of the parabolic form of the equations. 
\end{proof}

We have the following version of ``Laplace's method'' for our function $G(t)=\cosh (\sqrt{\frac{\gamma}{\lambda}}(T-t) )$ as $\lambda \to 0$.

\begin{lemma}
\label{lemma:IntegratedGConvergence}
Given $t\in[0,T)$ and a continuous function $F$ on $[t,T]$,
\begin{equation}\label{eq:laplace}
\sqrt{\frac{\gamma}{\lambda}}\int_t^T\left(\frac{-G^\prime(u)}{G(t)}\left(\int_t^uF(s)ds\right)\right)du \to F(t) \quad  \mbox{as $\lambda \to 0$.}
\end{equation}
\end{lemma}

\begin{proof}
The left-hand side of~\eqref{eq:laplace} can be decomposed as
\begin{align*}
\sqrt{\frac{\gamma}{\lambda}}\int_t^T\frac{-G^\prime(u)}{G(t)}\int_t^u\Big(F(s)-F(t)\Big)dsdu+\sqrt{\frac{\gamma}{\lambda}}\int_t^T\frac{-G^\prime(u)}{G(t)}\int_t^uF(t)dsdu.
\end{align*}
Using the uniform continuity of $F$ on $[t,T]$ and observing that $\sqrt{\frac{\gamma}{\lambda}}\frac{-G^\prime(\cdot)}{G(t)}$ converges to $0$ locally uniformly on $(t,T]$,
one verifies that the first term vanishes for $\lambda\to 0$. Integration by parts and $G=\frac{\lambda}{\gamma}G''$ show that the second term converges to $F(t)$.  
\end{proof}

Together with the uniform bounds from Proposition~\ref{proposition:ZeroOrderConvergenceforSmallTransactionCost}, Lemma~\ref{lemma:IntegratedGConvergence} allows us to compute the leading-order expansions of $v_i^\lambda-v^\lambda$ and its derivatives.

\begin{lemma}
\label{lemma:ConvergenceforFirstOrderCorrection}
For $k=0,1,2$ 
and $(t,x) \in [0,T)\times \mathbb{R}$, we have
\begin{equation}\label{eq:limit11}
\lim_{\lambda\to 0} \frac{\partial^k_xv_i^\lambda(t,x)-\partial^k_xv^\lambda(t,x)}{\sqrt{\lambda}} =\frac{\partial^k_x\mathcal{L}^i v^0(t,x)}{\sqrt{\gamma}}.
\end{equation}
\end{lemma}

\begin{proof}
The proof is similar for $k=0,1,2$; we only spell it out in the case $k=2$ for which the computations are most involved.
By Theorem~\ref{th:equilibriumPDE} and Remark~\ref{rk:PDEsolSmooth}, the second-order $x$-derivatives of the functions $v_i^\lambda$, $i=1,\ldots,N$ from Theorem~\ref{th:equilibriumPDE} satisfy the following PDEs obtained by differentiating~\eqref{eq:PDEvi} twice with respect to the spatial variable:
\begin{align*}
\partial_t&\partial_{xx}v^\lambda_i+\frac{1}{2}\sigma_i^2\partial_{xx}\partial_{xx}v^\lambda_i+(b_i+2\sigma_i\partial_x\sigma_i)\partial_x\partial_{xx}v^\lambda_i \\
&+\Big(c_i+\frac{G^\prime}{G}\Big)\partial_{xx}v_i^\lambda+\partial_{xx}b_i\partial_xv_i^\lambda-\frac{G^\prime}{G}\frac{1}{N}\sum_{j=1}^N\partial_{xx}v^\lambda_j=0, \quad \partial_{xx}v^\lambda_i(T,\cdot)=f^{\prime\prime},
\end{align*}
where $c_i=2\partial_xb_i+(\partial_x\sigma_i)^2+\sigma_i\partial_{xx}\sigma_i$. As all functions appearing here are bounded either by assumption or by Remark \ref{rk:PDEsolSmooth}, the Feynman--Kac formula and $G'/G=(\log G)'$ yield the stochastic representation
\begin{align*}
&\partial_{xx}v_i^\lambda(t,x) \\
&=E'_{t,x}\Bigg[\int_t^T\frac{-G^\prime(u)}{G(t)}\Big(e^{\int_t^uc_id\tau}\partial_{xx}v^\lambda(u,X_u)\Big)du+\frac{e^{\int_t^Tc_id\tau}f^{\prime\prime}(X_T)}{G(t)}\\
&\qquad\qquad+\int_t^T\frac{G(u)}{G(t)}\Big(e^{\int_t^uc_id\tau}\partial_{xx}b_i\partial_xv_i^\lambda(u,X_u)\Big)du\Bigg]
\end{align*}
where the expectation $E'[\cdot]$ is taken under the measure $Q'$ for which the state variable has dynamics $dX_t=(b_i+2\sigma_i\partial_x\sigma_i)(t,X_t)dt+\sigma_i(t,X_t)dW^i_t$. Together with $\int_t^T -\frac{G'(u)}{G(t)}du=1-\frac{1}{G(t)}$, this implies
\begin{align}
\frac{\partial_{xx}v_i^\lambda-\partial_{xx}v^\lambda}{\sqrt{\lambda}}=E'_{t,x}\Bigg[&\int_t^T\frac{-G^\prime(u)}{\sqrt{\lambda}G(t)}\Big(e^{\int_t^uc_id\tau}\partial_{xx}v^\lambda(u,X_u)-\partial_{xx}v^\lambda(t,x)\Big)du \notag \\
&+\frac{e^{\int_t^Tc_id\tau}f^{\prime\prime}(X_T)-\partial_{xx}v^\lambda(t,x)}{\sqrt{\lambda}G(t)} \label{eq:feynman}\\
&+\int_t^T\frac{G(u)}{\sqrt{\lambda}G(t)}\Big(e^{\int_t^uc_id\tau}\partial_{xx}b_i\partial_xv_i^\lambda(u,X_u)\Big)du\Bigg]. \notag
\end{align}
Recalling that $c_i$, $f^{\prime\prime}$ and (by Lemma~\ref{lemma:BoundforDerivativesofv}) also $\partial_{xx} v^\lambda$ are bounded (uniformly in $\lambda$), dominated convergence and the definition of $G$ show that the expectation of the second term on the right-hand side of~\eqref{eq:feynman} converges to zero as $\lambda\to 0$. In view of $\lim_{\lambda\to 0} \frac{G(T)}{\sqrt{\lambda}G(t)}=0$, dominated convergence and integration by parts show that the expectation of the third term converges to 
\begin{align}
&E'_{t,x}\left[ \lim_{\lambda\to 0} \int_t^T\frac{-G^\prime(u)}{\sqrt{\lambda}G(t)}\Big(\int_t^ue^{\int_t^sc_id\tau}\partial_{xx}b_i\partial_xv_i^\lambda(s,X_s)ds\Big)du\right] \notag\\
&\quad =\partial_{xx}b_i\partial_xv^0(t,x). \label{eq:fclimit2}
\end{align}
Here we have used Corollary~\ref{co:viConvergetov} and Proposition~\ref{proposition:ZeroOrderConvergenceforSmallTransactionCost}, and Lemma~\ref{lemma:IntegratedGConvergence} for the equality. Finally, the expectation of the first term on the right-hand side of~\eqref{eq:feynman} can be rewritten by applying It\^o's formula to $e^{\int_t^u c_i d\tau}\partial_{xx}v^\lambda(u,X_u)$, inserting the $Q'$-dynamics of $X$ and taking into account that the corresponding local martingale part has expectation zero because all involved functions are bounded. Dominated convergence as well as  Proposition~\ref{proposition:ZeroOrderConvergenceforSmallTransactionCost} and Lemma~\ref{lemma:IntegratedGConvergence} then show that the corresponding limit for $\lambda \to 0$ is  
\begin{align}
&E'_{t,x}\left[\lim_{\lambda\to 0} \int_t^T\frac{-G^\prime(u)}{\sqrt{\lambda}G(t)}\int_t^u e^{\int_t^s c_i d\tau} \mathcal{L}_i^{\prime\prime}\partial_{xx} v^\lambda(s,X_s)ds du\right] = \mathcal{L}_i^{\prime\prime}\partial_{xx}v^0(t,x), \label{eq:fclimit3}
\end{align}
where $\mathcal{L}_i^{\prime\prime}=\partial_t+\frac{1}{2}\sigma_i^2\partial_{xx}+(b_i+2\sigma_i\partial_x\sigma_i)\partial_x+c_i \rm{Id}$. The assertion for $k=2$ now follows from (\ref{eq:feynman}--\ref{eq:fclimit3}) by observing that $\mathcal{L}_i^{\prime\prime}\partial_{xx}+\partial_{xx}b_i\partial_x=\partial_{xx}\mathcal{L}^i$.
\end{proof}

We can now prove the expansion of the equilibrium price for small transaction costs.

\begin{proof}[Proof of Theorem~\ref{theorem:TransactionCostExpansion}]
We first observe that the PDE~\eqref{equ:Differenceofvv0} for $v^\lambda-v^0$ admits the Feynman--Kac representation
\begin{align*}
&(v^\lambda-v^0)(t,x)\\
&=\frac{1}{N}\sum_{i=1}^N \bar E_{t,x}\left[\int_t^T\Big(\frac{1}{2}\sigma_i^2(\partial_{xx}v_i^\lambda-\partial_{xx}v^\lambda)+b_i (\partial_xv_i^\lambda-\partial_xv^\lambda)\Big)(s,X_s)ds\right].
\end{align*}
Dominated convergence and the limits computed in Lemma~\ref{lemma:ConvergenceforFirstOrderCorrection} then yield
\begin{align}
\lim_{\lambda \to 0} \frac{v^\lambda - v^0}{\sqrt{\lambda}}(t,x) & = \sqrt{\frac{1}{\gamma}} \frac1N \sum_{i=1}^N \bar E_{t,x}\left[\int_t^T \Big(\frac{1}{2}\sigma_i^2\partial_{xx}+b_i  \partial_x\Big)\mathcal{L}^i v^0(s,X_s)ds\right] \notag \\
&= \sqrt{\frac{1}{\gamma}} \frac1N \sum_{i=1}^N \bar E_{t,x}\left[\int_t^T \Big(\mathcal{L}^i -\partial_t\Big)\gamma\hat\phi^{i,0}(s,X_{s}) ds\right] \label{eq:limit1}
\end{align}
where $\hat\phi^{i,0}=\cL^{i}v^0/\gamma$ is the frictionless equilibrium portfolio function of agent~$i$; cf.~\eqref{eq:phi0}. As these strategies clear the market, the sum of their time derivatives is zero and the pointwise limit~\eqref{eq:limit1} simplifies to \eqref{eq:ve}. The family $\{\lambda^{-1/2}(v^\lambda - v^0)\}_{\lambda>0}$ is bounded and equicontinuous by Proposition~\ref{proposition:ZeroOrderConvergenceforSmallTransactionCost}; whence, the convergence is in fact locally uniform as a consequence of the Arzel\`a--Ascoli theorem.
\end{proof}

\section{Asymptotics for Small Holding Costs}\label{se:holdingAsymps}

Next, we study the asymptotics of the equilibrium price from Theorem~\ref{th:equilibriumPDE} for small holding costs $\gamma\to 0$ (and fixed transaction costs $\lambda>0$). To emphasize the dependence on $\gamma$, we denote the price $v$ by $v^\gamma$ in this section. We again focus on the case of a one-dimensional state variable ($d=1$) with smooth drift and diffusion coefficients and terminal condition; see Remark~\ref{rk:PDEsolSmooth}.

To formulate the result, we first note that the risk-neutral version $\gamma=0$ of our model is well posed and essentially covered as a simple special case of Theorem~\ref{th:equilibriumPDE} (with the same proof, read with the conventions $G(u)/G(s)=1$ and $G'(u)/G(s)=0$). The corresponding equilibrium price is the average of all agents' conditional expectations,
\begin{equation}\label{eq:rnprice}
v^0(t,x)=\frac{1}{N} \sum_{i=1}^N v^0_i(t,x)=\frac{1}{N} \sum_{i=1}^N E^i_{t,x}[f(X_T)],
\end{equation}
and the corresponding portfolios are
\begin{equation}\label{eq:rnstrat}
\phi^{i,0}_t= a_i +\int_0^t E^i_{s}\left[\int_s^T \frac{\mathcal{L}^i v^0(u,X_u)}{\lambda} du \right]ds.
\end{equation}
(The above notation for the case $\gamma=0$ should not be confused with the notation for the case $\lambda=0$ in the preceding section.) 

Lemma~\ref{le:optPortfolio} shows that when $\gamma>0$ and $\lambda>0$, the optimal portfolios take into account future expected returns that are discounted with a kernel determined by $\gamma/\lambda$. As a limiting case, we have seen that the no-transaction-cost portfolio~\eqref{eq:phi0} for $\lambda=0$ only takes into account the current (subjective) drift rates; this corresponds to an infinite discount. In the opposite extreme, the no-holding-cost portfolio~\eqref{eq:rnstrat} aggregates the future expected returns without discounting.

Accordingly, we expect small holding costs to play a similar role as large transaction costs. Indeed, Theorem~\ref{th:equilibriumPDE} shows that when the supply $a_0$ vanishes, the equilibrium price only depends on the ratio $\gamma/\lambda$---the ``urgency parameter'' that determines optimal execution trajectories~\cite{almgren.chriss.01} and, more generally, optimal trading strategies with transaction costs in various contexts; cf., e.g., \cite{moreau.al.17} and  the references therein. 
When $a_{0}>0$, Theorem~\ref{th:equilibriumPDE} shows that the equilibrium price is 0-homogeneous in $(\gamma, \lambda,1/a_{0})$. This means that the asset price remains invariant if the inverse of the supply is rescaled in the same manner as transaction and holding costs: the larger trading and holding costs of bigger asset positions are offset by reduced friction coefficients.

The main result of this section is the following regular perturbation expansion for small holding costs $\gamma\to 0$. 

\begin{theorem}
\label{theorem:HoldingCostExpansion}
For fixed transaction costs $\lambda>0$ and small holding costs $\gamma \to 0$, the equilibrium price function from Theorem~\ref{th:equilibriumPDE} has the expansion
\begin{equation}\label{eq:regular}
v^\gamma(t,x)=v^0(t,x)+\gamma v^{*}(t,x)+o(\gamma) \quad \mbox{uniformly on $[0,T] \times \mathbb{R}$.}
\end{equation}
Here $v^{0}$ is the equilibrium price~\eqref{eq:rnprice} for $\gamma=0$ and
\begin{align*}
v^{*}(t,x) = -\frac{1}{N}\sum_{i=1}^N E^i_{t,x}\left[\int_t^T\phi_s^{i,0}ds\right]
\end{align*}
where $\phi_t^{i,0}$ is the optimal strategy~\eqref{eq:rnstrat} of agent~$i$ for $\gamma=0$ and the expectation is taken under agent~$i$'s belief~$Q_{i}$.
\end{theorem}

The reference point for the expansion~\eqref{eq:regular} is the risk-neutral price $v^{0}$ of~\eqref{eq:rnprice}. In this limiting case, agents only consider future expected returns. Other things equal, agents reduce the magnitude of their positions when holding costs are introduced.
The above expression for $v^{*}$ reflects each agent's expectation $E^i_{t,x}[\int_t^T\phi_s^{i,0}ds]$ of their average future position. Adding holding costs reduces the demand by agents who expect to be long on average, and the converse holds for shorts. The resulting sign of the price correction will thus depend on the aggregate expectations in the market. Indeed, 
the formula for~$v^{*}$ shows that at the first order, the arithmetic average over all agents' expected average positions is the negative of the correction.

\begin{proof}[Proof of Theorem~\ref{theorem:HoldingCostExpansion}]
\emph{Step~1}. Similarly as for Theorem~\ref{theorem:TransactionCostExpansion}, the first step towards proving this expansion is to establish that the functions $v^\gamma_i$ from Theorem~\ref{th:equilibriumPDE} are uniformly bounded in~$\gamma$. Indeed, note that the function $\frac{\lambda G^\prime(t)^2}{NG(t)^2}a_0$ is bounded locally uniformly in~$\gamma$. Hence, Lemma~\ref{lemma:UniformBound} applied with the PDEs (\ref{eq:PDEvi}--\ref{eq:PDEv}) from Theorem~\ref{th:equilibriumPDE} yields that given $0<\bar{\gamma}<\infty$, there exists $M>0$ such that
\begin{equation}\label{eq:UniformlyBDDforPricewithSmallHoldingCost}
  |v^\gamma_i(t,x)|\leq M  \quad \mbox{for all $\gamma \in [0,\bar{\gamma}]$.}
\end{equation} 

\emph{Step~2}. Next, we show that as $\gamma \to 0$,
\begin{align}\label{eq:ZeroOrderConvergenceforSmallHoldingCost}
|v_{i}^\gamma(t,x)-v_{i}^0(t,x)| &\to 0 \quad \mbox{uniformly on $[0,T]\times\mathbb{R}$.}
\end{align}
Indeed, (\ref{eq:PDEvi}--\ref{eq:PDEv}) show that 
\begin{align*}
\partial_t(v^\gamma_i-v^0_i)&+\frac{1}{2}\sigma_i^2\partial_{xx}(v^\gamma_i-v^0_i)+b_i\partial_x(v^\gamma_i-v^0_i)\\
&+\frac{G^\prime(t)}{G(t)}\left(v_i^\gamma-\frac{1}{N}\sum_{j=1}^Nv^\gamma_j-\frac{\lambda G^\prime(t)}{NG(t)}a_0\right)=0, \quad (v^\gamma_i-v^0_i)(T,\cdot)=0.
\end{align*}
Thus, the Feynman--Kac formula yields
\begin{align}
\label{equ:viDifferencewithGamma}
&(v^\gamma_i-v^0_i)(t,x)\notag\\
&\quad =E^i_{t,x}\left[\int_t^T\frac{G^\prime(s)}{G(s)}\Big(v_i^
\gamma(s,X_s)-\frac{1}{N}\sum_{j=1}^Nv^\gamma_j(s,X_s)-\frac{\lambda G^\prime(s)}{NG(s)}a_0\Big)ds\right]
\end{align}
where the expectation is taken under agent $i$'s subjective probability measure~$Q_{i}$. Note that $G^\prime(t) \to 0$ and $G(t) \to 1$ as $\gamma \to 0$, uniformly on $[0,T]$. In view of~\eqref{eq:UniformlyBDDforPricewithSmallHoldingCost}, we conclude~\eqref{eq:ZeroOrderConvergenceforSmallHoldingCost}.

\emph{Step~3}. We can now prove the expansion from Theorem~\ref{theorem:HoldingCostExpansion}.
By~\eqref{equ:viDifferencewithGamma},
\begin{align*}
&\frac{(v^\gamma_i-v^0_i)(t,x)}{\gamma}\\
&\quad =E^i_{t,x}\left[\int_t^T\frac{G^\prime(s)}{\gamma G(s)}\Big(v_i^
\gamma(s,X_s)-\frac{1}{N}\sum_{j=1}^Nv^\gamma_j(s,X_s)-\frac{\lambda G^\prime(s)}{NG(s)}a_0\Big)ds\right].
\end{align*}
Using the definition of $G$ and Dini's theorem,  
\begin{equation}\label{eq:Glimit}
\lim_{\gamma \to 0} \frac{G^\prime(t)}{G(t)}=0 \quad \mbox{and} \quad \lim_{\gamma \to 0}\frac{G^\prime(t)}{\gamma G(t)}=-\frac{T-t}{\lambda}, \quad  \mbox{uniformly on $[0,T]$.}
\end{equation}
Together with~\eqref{eq:UniformlyBDDforPricewithSmallHoldingCost}, dominated convergence, \eqref{eq:ZeroOrderConvergenceforSmallHoldingCost} and \eqref{eq:rnprice}, this yields
\begin{align*}
\lim_{\gamma \to 0}\frac{v_i^\gamma(t,x)-v_i^0(t,x)}{\gamma} &=E^i_{t,x}\left[\int_t^T\frac{T-s}{\lambda}\Big(v^0(s,X_s)-v_i^0(s,X_s)\Big)ds\right]
\end{align*}
uniformly on $[0,T]\times \mathbb{R}$. In view of the definition of $v^\gamma$ in~\eqref{eq:PDEv}, and \eqref{eq:Glimit}, it follows that
\begin{align}
&\lim_{\gamma \to 0}\frac{v^\gamma(t,x)-v^0(t,x)}{\gamma} \notag\\
&\quad =\frac{1}{N}\sum_{i=1}^NE^i_{t,x}\left[\int_t^T\frac{T-s}{\lambda}\Big(v^0(s,X_s)-v_i^0(s,X_s)\Big)ds\right]-\frac{(T-t)a_0}{N}. \label{eq:penult}
\end{align}
By~\eqref{eq:rnstrat} and the first identity of~\eqref{eq:intByPartPDEproof} in the special case $\gamma=0$, we have
\begin{align*}
\phi^{i,0}_t - a_{i} &=  \int_0^tE^i_{s}\left[\frac{1}{\lambda}\Big(f(X_T)-v^0(s,X_s)\Big)\right]ds\\
&=  \int_0^t\frac{1}{\lambda}\Big(v_i^0(s,X_s)-v^0(s,X_s)\Big)ds.
\end{align*}
Using this identity to integrate~\eqref{eq:penult} by parts and taking into account the market-clearing condition $\sum_{i=1}^N \phi^{i,0} =a_0$, the theorem follows. 
\end{proof}

\section{Example: Mean-Reversion Trading}\label{se:example}

To gain further intuition for the equilibrium of Theorem~\ref{th:equilibriumPDE}, we consider an example that can be solved explicitly up to a system of linear ODEs. We will also use this example to test the numerical accuracy of the expansions for small transaction and holding costs relative to the exact solution.  

Suppose that $f(x)=x$, so that at time $T$, the state $X$ represents the asset's payoff. Agents believe that $X$ has mean-reverting dynamics
\begin{equation}\label{eq:OUdyn}
dX_t= \kappa_i(\bar{X}- X_t) dt +\sigma dW^i_t.
\end{equation}
That is, agents agree on the volatility $\sigma>0$ and the mean-reversion level $\bar{X}>0$, but disagree about the mean-reversion speed $\kappa_i>0$. This can be interpreted as a simple model for a forward contract on a mean-reverting underlying such as an FX rate.
As is natural in that context, and to simplify the exposition, we henceforth assume that the net supply of the contract is $a_0=0$.

\subsection{Equilibrium with Costs}\label{ss:costex}

We first consider the exact equilibrium price $v$ with transaction costs $\lambda>0$ and holding costs $\gamma>0$ from Theorem~\ref{th:equilibriumPDE}. For the linear state dynamics~\eqref{eq:OUdyn}, the parabolic system (\ref{eq:PDEvi}--\ref{eq:PDEterminalCond}) can be reduced to a system of linear ODEs by the ansatz 
\begin{equation*}
v_i^\lambda(t,x)=A_i(t)+B_i(t)x, \quad i=1,\ldots,N.
\end{equation*}
Indeed, writing $\mathbbm{1}_N$ and $I_N$ for the $N\times N$-matrices of ones and the identity matrix, respectively, the deterministic functions $B=(B_1,\ldots,B_N)^\top$ and $A=(A_1,\ldots,A_N)^\top$ satisfy
\begin{align*}
B'(t) &=\left[ \mathrm{diag}(\kappa_1,\ldots,\kappa_n)+\frac{G'(t)}{G(t)}\left(\frac1N \mathbbm{1}_N-I_N\right)\right]B(t), \notag\\
B(T)&=1
\end{align*}
and
\begin{align*}
A'(t) &= \frac{G'(t)}{G(t)}\left(\frac1N \mathbbm{1}_N-I_N\right)A(t)-\bar{X} \mathrm{diag}(\kappa_1,\ldots,\kappa_N)B(t)
,\notag\\
A(T)&=0. 
\end{align*}
These ODEs have unique, smooth solutions. 
Moreover, the equilibrium price then satisfies
\begin{equation}\label{eq:eqfricex}
  v(t,x)= \frac1N \sum_{i=1}^N \left(A_i(t) +B_i(t)x\right) = \bar{X}+(x-\bar{X})\frac1N \sum_{i=1}^N B_i(t),
\end{equation}
where we have used the ODEs for the $A_i$ and $B_i$ for the second equality. To be precise, the unbounded terminal conditions and state dynamics~\eqref{eq:OUdyn} do not satisfy the boundedness assumptions of Theorem~\ref{th:equilibriumPDE}. However, with the unique solutions $A$ and $B$ of the above ODEs at hand, the arguments in the proof of Theorem~\ref{th:equilibriumPDE} show that~\eqref{eq:eqfricex} identifies the unique equilibrium price in the class from smooth functions with linear growth, say.

\subsection{Transaction-Cost  Asymptotics}\label{ss:extcasymp}

We first study the equilibrium $v^{0}$ with vanishing transactions costs $\lambda=0$ and fixed holding costs $\gamma>0$. As the state variable has the dynamics
\begin{equation}\label{eq:barX}
dX_t=\bar\kappa (\bar{X}-X_t) dt +\sigma dW_t \quad \mbox{with} \quad\bar{\kappa}=\frac1N \sum_{i=1}^N \kappa_i
\end{equation}
under the aggregate measure $\bar{Q}$, Proposition~\ref{pr:NoTransactionCostPrice} with $a_{0}=0$ yields that
 \begin{align}\label{eq:v0ex}
v^0(t,x) &= \bar{E}_{t,x}\left[X_T\right] = \bar{X}+(x-\bar{X})e^{-\bar\kappa (T-t)}.
\end{align}
As a result, agent $i$ believes that the frictionless equilibrium price has dynamics
\begin{align}
dv^0(t,X_t) &= (\kappa_i-\bar{\kappa})e^{-\bar{\kappa}(T-t)}(\bar{X}-X_t)dt+ e^{-\bar\kappa(T-t)}\sigma dW^i_t \notag\\
&= (\kappa_i-\bar{\kappa})\big(\bar{X}-v^0(t,X_t)\big)dt+ e^{-\bar\kappa(T-t)}\sigma dW^i_t. \label{eq:fldyn}
\end{align}
This means that agents who believe in faster than average mean-reversion (i.e., $\kappa_{i}>\bar\kappa$) observe a mean-reverting process. By contrast, agents who believe in slower than average mean reversion conclude that the process exhibits ``momentum'' in that prices above the mean-reversion level are followed by further positive drifts. Whence, in equilibrium, the market is endogenously populated by both ``mean-reversion traders'' and ``trend-followers'' even though all agents believe that the underlying has a mean-reverting fundamental value. 

Next, we study the leading-order correction $v^\lambda(t,x)-v^0(t,x)$ for $\lambda\to0$. Again, Theorem~\ref{theorem:TransactionCostExpansion} does not apply directly due to the unbounded coefficients, but it is straightforward to carry out the arguments in the proof for the example at hand. 
%
%
%
%
%
%
%
%
Thus, the leading-order correction is $\sqrt{\lambda}v^{*}(t,x)$ with
\begin{align}
v^{*}(t,x)&= \frac{1}{\sqrt{\gamma}N} \sum_{i=1}^N \bar{E}_{t,x}\left[\int_t^T e^{-\bar\kappa (T-s)}(\bar{\kappa}-\kappa_i)^{2}(X_s-\bar{X}) ds\right] \notag\\
&=\frac{1}{\sqrt{\gamma}} \Bigg(\frac{1}{N}\sum_{i=1}^N (\bar{\kappa}-\kappa_{i})^2\Bigg) \left[\int_t^T \bar{E}_{t,x}[X_s-\bar{X}] e^{-\bar\kappa (T-s)} ds\right] \notag\\
&=\frac{1}{\sqrt{\gamma}} \Bigg(\frac{1}{N}\sum_{i=1}^N (\bar{\kappa}-\kappa_{i})^2\Bigg)\left[\int_t^T e^{-\bar\kappa (T-t)} (x-\bar{X}) ds\right] \notag\\
&=\sqrt{\frac{1}{\gamma}} \Bigg(\frac{1}{N}\sum_{i=1}^N (\bar{\kappa}-\kappa_{i})^2\Bigg) (T-t) e^{-\bar\kappa (T-t)} (x-\bar{X}). \label{eq:exptcasymp}
\end{align}

Note that $\partial_{x}v^{*}\geq0$, so that the equilibrium volatility is always increased when small transaction costs are added. This is in line with the asymmetric information model of~\cite{danilova.juillard.19}, the risk-sharing model of~\cite{herdegen.al.19}, numerical results of~\cite{adam.al.15,buss.al.16}, and empirical studies such as~\cite{hau.06,jones.seguin.97,umlauf.93}. 

In our model, the reason for the increased volatility is that the sign of the correction term $v^{*}$ is determined by $x-\bar{X}$, so that transaction costs amplify the fluctuations of the frictionless equilibrium price~\eqref{eq:v0ex}. Let us now discuss why illiquidity affects price levels in this manner. In view of the above formula for $v^{*}$, adding small transaction costs increases equilibrium prices when $X_t>\bar{X}$ and reduces prices for $X_t<\bar{X}$. If $X_t>\bar{X}$, agents who believe in larger than average mean-reversion speeds predict the frictionless equilibrium price~\eqref{eq:fldyn} to mean-revert downwards towards its long-run mean. Conversely, agents believing in a lower than average mean-reversion speed expect the positive trend to continue and prices to rise even further. Accordingly, the first group of agents wants to sell the asset and the second group wants to purchase it. With small transaction costs added, these trading motives persist, yet changes in portfolios can only be implemented gradually. Accordingly, agents do not only take into account the difference between the current value of the state variable and its long-run mean, but also their expected differences in the future. Since agents believing in faster mean-reversion expect differences to disappear faster, they have a weaker motive to act on the trading opportunities they observe. For $X_t>\bar{X}$, this means that sellers have a weaker motive to trade than buyers, so that prices need to rise in order to clear the market. For $X_t <\bar{X}$, the situation is reversed and small transaction costs decrease prices relative to their frictionless counterparts. 

In summary, adding small transaction costs increases prices above the mean-reversion level and decreases price below it, thereby generating larger price fluctuations and a larger equilibrium volatility. As optimists and pessimists are similarly affected by the costs, the effect of illiquidity on price levels in our model is ambiguous. Depending on the situation, an increase of $\lambda$ can lead to an ``illiquidity discount'' as observed e.g.\ in \cite{amihud.mendelson.86a}, or it may increase the price as in \cite{DuffieGarleanuPedersen.01} where illiquidity can be an obstruction to shorting. One can note that in our example, the correction term $v^{*}$ mean-reverts around zero under each agent's probability measure. In that sense, the average price level remains unchanged.

\subsection{Holding-Cost Asymptotics}\label{ss:exhcasymp}


We now turn to the small-holding-cost asymptotics from Section~\ref{se:holdingAsymps}. As a first step, we compute the equilibrium price $v^{0}$ with vanishing holding costs~$\gamma=0$ and fixed transactions costs $\gamma>0$. From \eqref{eq:rnprice}, we have
$$v^0(t,x)=\frac{1}{N} \sum_{i=1}^N v^0_i(t,x),$$
where
\begin{equation}\label{eq:v0iex}
v^0_i(t,x)=E^i_{t,x}[X_T]=\bar{X}+(x-\bar{X})e^{-\kappa_i (T-t)}.
\end{equation}
By It\^o's formula, agent $i$ believes that this risk-neutral equilibrium price has dynamics
\begin{align*}
&dv^0(t,X_t)\\
&= \frac{1}{N}\sum_{j=1}^N(\kappa_j-\kappa_i)e^{-\kappa_j(T-t)}(X_t-\bar{X})dt+\frac{1}{N}\sum_{j=1}^Ne^{-\kappa_j(T-t)}\sigma dW_t^i\\
&=\left(\kappa_i-\frac{\sum_{j=1}^N\kappa_je^{-\kappa_j(T-t)}}{\sum_{j=1}^Ne^{-\kappa_j(T-t)}}\right)\left(\bar{X}-v^0(t,X_t)\right)dt+\frac{1}{N}\sum_{j=1}^Ne^{-\kappa_j(T-t)}\sigma dW_t^i.
\end{align*}
The first factor is the difference between $\kappa_{i}$ and a (time-dependent) weighted average of $\kappa_{1},\dots,\kappa_{N}$. Thus, the interpretation is similar as for the equilibrium~\eqref{eq:fldyn} with $\lambda=0$: agents who believe in fast mean reversion observe a mean-reverting asset price whereas agents believing in slow mean reversion perceive momentum. One can also note that the equilibrium volatility without holding costs is always larger than or equal to its counterpart without transaction costs. This follows by applying Jensen's inequality to the gradients of~$v^{0}$ and~\eqref{eq:v0ex}.

We now turn to the leading-order correction term for $\gamma\to0$. 
Again, the boundedness assumptions in Theorem~\ref{theorem:HoldingCostExpansion} are not satisfied in this example, but the arguments in the proof still apply.
Accordingly, using the representation~\eqref{eq:penult}, we have
\begin{align}
&v^{*}(t,x) \label{eq:exexpansion}\\
&=\frac{1}{N}\sum_{i=1}^NE^i_{t,x}\left[\int_t^T\frac{T-s}{\lambda}\Big(v^0(s,X_s)-v_i^0(s,X_s)\Big)ds\right]\notag \\
&=\frac{1}{N}\sum_{i=1}^N\int_t^T\frac{T-s}{\lambda}\Big(\frac{1}{N}\sum_{j=1}^Ne^{-\kappa_j (T-s)}-e^{-\kappa_i (T-s)}\Big)e^{-\kappa_i (s-t)}(x-\bar{X})ds \notag\\
&=\frac{(x-\bar{X})(T-t)}{\lambda N^2}\left(\sum_{i\neq j}\frac{1}{\kappa_i-\kappa_j}e^{-\kappa_j (T-t)}-\frac{(T-t)(N-1)}{2}\sum_{i=1}^N e^{-\kappa_i (T-t)}\right) \notag
\end{align}
where the last equality follows from an elementary but lengthy integration.

The Chebychev sum inequality applied to the second representation shows that the coefficient multiplying $x-\bar{X}$ is always negative. Whence, adding small holding costs increases the risk-neutral equilibrium price when the state process $X_t$ is below its mean-reversion level $\bar{X}$ and decreases it when $X_t>\bar{X}$. Since larger holding costs play the same role as lower transaction costs in our model for $a_0=0$, the intuition for this is the converse of the argument for adding small transaction costs in Section~\ref{ss:extcasymp}. 

In  particular, in view of~\eqref{eq:v0iex}, adding small holding costs dampens the fluctuations of the risk-neutral equilibrium price and accordingly reduces the equilibrium volatility.  This negative effect on the equilibrium volatility and the positive effect of small transaction costs are consistent with the observation made above that the equilibrium volatility without transaction costs always lies below its counterpart with no holding costs. In fact, the exact equilibrium volatility $\frac1N \sum_{i=1}^N B_i(t)$ from Section~\ref{ss:costex} smoothly interpolates between these two extreme cases as $\gamma/\lambda$ ranges between $\infty$ and~$0$. 

\subsection{A Calibrated Example}

To assess the accuracy of the small-cost asymptotics from Sections~\ref{ss:extcasymp} and \ref{ss:exhcasymp}, we now compare the explicit asymptotic formulas~\eqref{eq:exptcasymp} and \eqref{eq:exexpansion} to the numerical solutions of the ODEs from Section~\ref{ss:costex} describing the exact equilibrium price. Throughout, we consider a time horizon of $T=3$ years.

To obtain reasonable values for the other model parameters, we calibrate the state dynamics~\eqref{eq:barX} to USD/EUR exchange rate data from 2009--2019 available from the website of the St.~Louis Fed at \url{https://fred.stlouisfed.org/series/DEXUSEU}. The model parameters can then be estimated by matching the first two stationary moments to their empirical counterparts and fitting the (linear) log-autocorrelation function to the empirical one using linear regression. This leads to
$$
\sigma=0.128, \quad \bar{X}=1.25, \quad \bar{\kappa}=0.575.
$$
With a zero net supply, this suffices to pin down the equilibrium price without transaction costs~\eqref{eq:v0ex}, since the latter does not depend on the agents' holding costs in this case. For the equilibrium prices with transaction costs~\eqref{eq:eqfricex}, we additionally need to specify each agent's individual belief as well as the transaction cost $\lambda$ and the holding cost $\gamma$. Inspired by similar parameter values used for commodities and equities in~\cite{garleanu.pedersen.13,cartea.jaimungal.16}, respectively, we use 
$$
\lambda=10^{-7} \quad \mbox{and} \quad  \gamma=10^{-8}.
$$ 
The free parameter $\kappa_1 = 2\bar{\kappa}-\kappa_2$ can in turn be chosen arbitrarily to capture the agents' disagreement about the mean-reversion speed of the exchange rate. For
$$\kappa_1= 3\kappa_2= 0.8625,$$
and $x=1$, equilibrium asset prices and volatilities are plotted in Figure~\ref{figA}. These numerical examples clearly display the qualitative properties derived from the asymptotic formulas in the previous sections. Indeed, the equilibrium values with both holding and transaction costs always lie between the limiting cases where only one of these costs is present. The corresponding volatility is increased by the trading cost, in line with the discussion in Section~\ref{ss:extcasymp}, and for $X_t=x<\bar{X}$ the equilibrium price with transaction costs lies below its no-transaction cost counterpart. 

\begin{figure}[t]
 \centering
   \includegraphics[width=.49\linewidth]{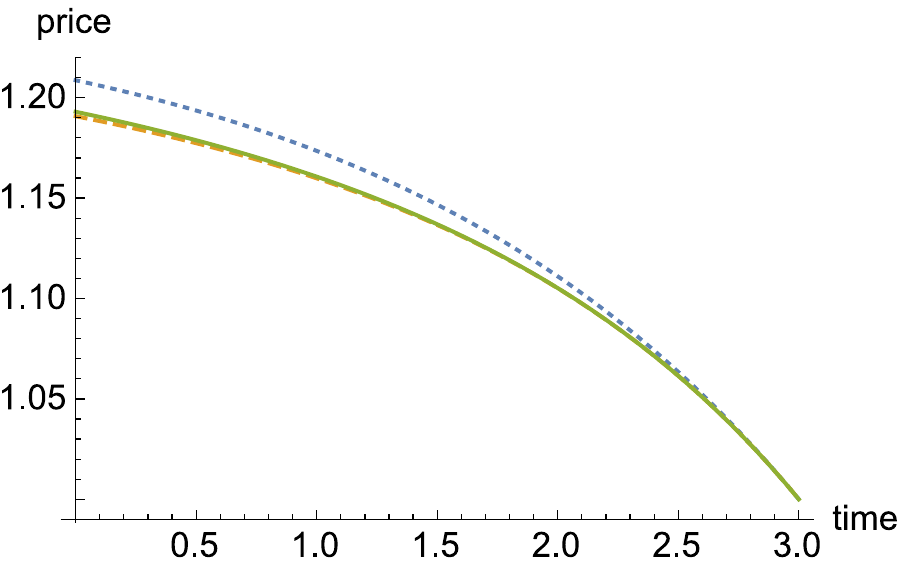}
   \includegraphics[width=.49\linewidth]{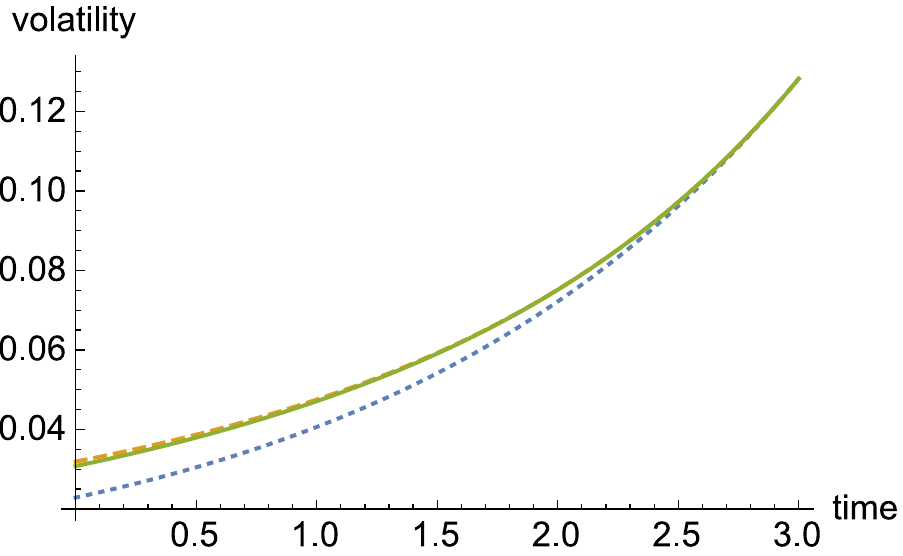}
\caption{Equilibrium prices (left) and volatilities (right) with both transaction and holding cost (solid), no transaction costs (dotted), and no holding costs (dashed).}
\label{figA}
\vspace{-2.5em}
\end{figure}

 \begin{figure}[b]
  \centering
  \includegraphics[width=.49\linewidth]{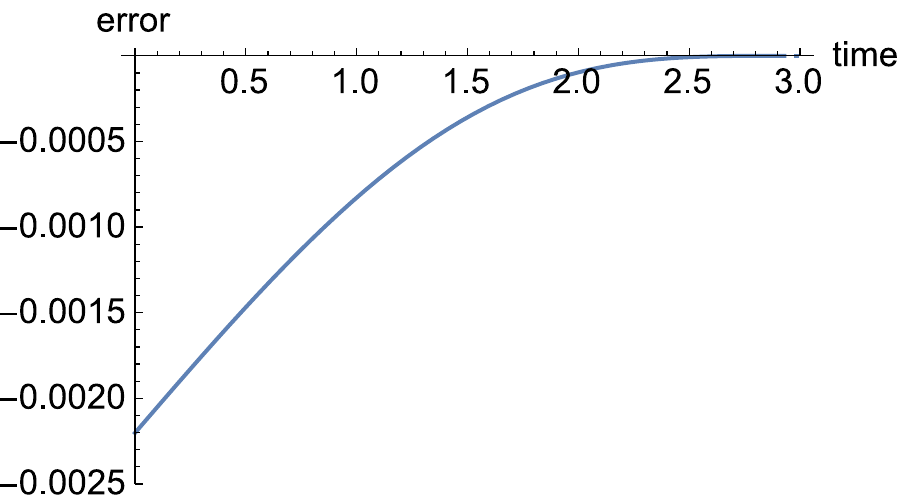}\;
  \includegraphics[width=.45\linewidth]{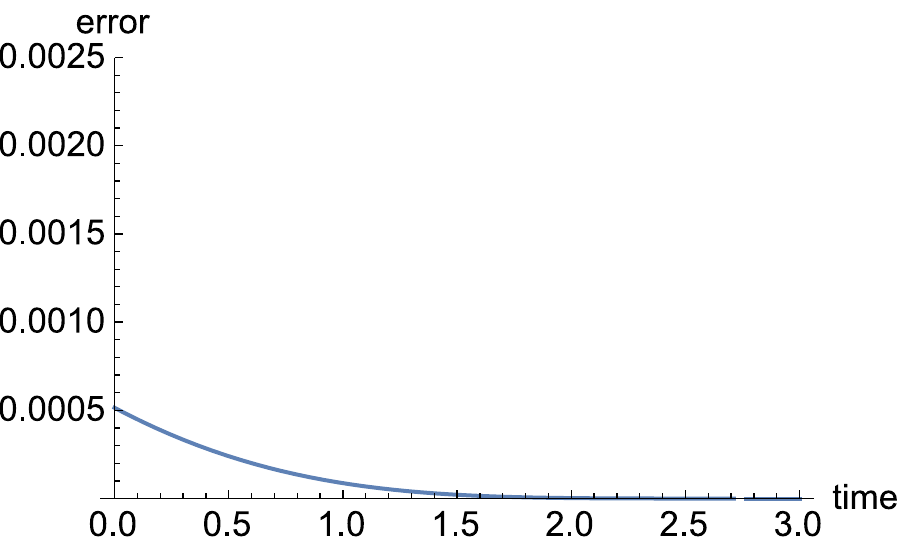}
\caption{Approximation errors $v^0-v^\gamma$ (left) and $v^0+\gamma v^*-v^\gamma$ (right).}
\label{figB}
\end{figure}

Figure~\ref{figA} also clearly shows that the equilibrium price with holding and transaction costs is much closer to the risk-neutral price than to the frictionless one. This is not surprising, since $\gamma/\lambda=0.1$ in this example. Accordingly, even though the qualitative predictions of both small-cost expansions are correct,  only the small-holding-cost expansion provides useful quantitative approximations here. As shown in Figure~\ref{figB}, using the first-order correction term $v^0+\gamma v^*-v^\gamma$ reduces the already small approximation error $v^0-v^\gamma$ by another order of magnitude.

\bibliography{stochfin}

\newcommand{\dummy}[1]{}
\begin{thebibliography}{10}

\bibitem{adam.al.15}
K.~Adam, J.~Beutel, A.~Marcet, and S.~Merkel.
\newblock Can a financial transaction tax prevent stock price booms?
\newblock {\em J. Mon. Econ.}, 76:S90--S109, 2015.

\bibitem{almgren.12}
R.~Almgren.
\newblock Optimal trading with stochastic liquidity and volatility.
\newblock {\em SIAM J. Financ. Math.}, 3(1):163--181, 2012.

\bibitem{almgren.chriss.01}
R.F. Almgren and N.~Chriss.
\newblock Optimal execution of portfolio transactions.
\newblock {\em J. Risk}, 3(2):5--39, 2001.

\bibitem{amihud.mendelson.86a}
Y.~Amihud and H.~Mendelson.
\newblock Asset pricing and the bid-ask spread.
\newblock {\em J. Financ. Econ.}, 17(2):223--249, 1986.

\bibitem{annkirchner.kruse.15}
S.~Ankirchner and T.~Kruse.
\newblock Optimal position targeting with stochastic linear--quadratic costs.
\newblock {\em Banach Center Publ.}, 104(1):9--24, 2015.

\bibitem{bank.al.17}
P.~Bank, H.M. Soner, and M.~Vo{\ss}.
\newblock Hedging with temporary price impact.
\newblock {\em Math. Financ. Econ.}, 11(2):215--239, 2017.

\bibitem{BayraktarMunk.17}
E.~Bayraktar and A.~Munk.
\newblock Mini-flash crashes, model risk, and optimal execution.
\newblock {\em Market Microstructure and Liquidity}, 4(1):1850010, 2018.

\bibitem{BechererSchweizer.05}
D.~Becherer and M.~Schweizer.
\newblock Classical solutions to reaction-diffusion systems for hedging
  problems with interacting {I}t{\^{o}} and point processes.
\newblock {\em Ann. Appl. Probab.}, 15(2):1111--1144, 2005.

\bibitem{Besala.63}
P.~Besala.
\newblock On solutions of {F}ourier's first problem for a system of non-linear
  parabolic equations in an unbounded domain.
\newblock {\em Ann. Polon. Math.}, 13:247--265, 1963.

\bibitem{BhamraUppal.14}
H.~S. Bhamra and R.~Uppal.
\newblock Asset prices with heterogeneity in preferences and beliefs.
\newblock {\em Rev. Financ. Stud.}, 27(2):519--580, 2014.

\bibitem{bonelli.al.18}
M.~Bonelli, A.~Landier, G.~Simon, and D.~Thesmar.
\newblock The capacity of trading strategies.
\newblock {\em Preprint HEC Paris}, 2019.

\bibitem{BouchardFukasawaHerdegenKarbe.18}
B.~Bouchard, M.~Fukasawa, M.~Herdegen, and J.~Muhle-Karbe.
\newblock Equilibrium returns with transaction costs.
\newblock {\em Finance Stoch.}, 22(3):569--601, 2018.

\bibitem{buss.dumas.19}
A.~Buss and B.~Dumas.
\newblock The dynamic properties of financial-market equilibrium with trading
  fees.
\newblock {\em J. Finance}, 74(2):795--844, 2019.

\bibitem{buss.al.16}
A.~Buss, B.~Dumas, R.~Uppal, and G.~Vilkov.
\newblock The intended and unintended consequences of financial-market
  regulations: A general-equilibrium analysis.
\newblock {\em J. Mon. Econ.}, 81:25--43, 2016.

\bibitem{cartea.jaimungal.16}
{\'A}.~Cartea and S.~Jaimungal.
\newblock A closed-form execution strategy to target volume weighted average
  price.
\newblock {\em SIAM J. Financial Math.}, 7(1):760--785, 2016.

\bibitem{casgrain.jaimungal.18}
P.~Casgrain and S.~Jaimungal.
\newblock Mean-field games with differing beliefs for algorithmic trading.
\newblock {\em Preprint arXiv:1810.06101v1}, 2018.

\bibitem{choi.al.18}
J.-H. Choi, K.~Larsen, and D.J. Seppi.
\newblock Equilibrium effects of {TWAP} and {VWAP} order splitting.
\newblock {\em Preprint SSRN:3146658}, 2019.

\bibitem{cvitanic.al.11}
J.~Cvitani{\'c}, E.~Jouini, S.~Malamud, and C.~Napp.
\newblock Financial markets equilibrium with heterogeneous agents.
\newblock {\em Rev. Finance}, 16(1):285--321, 2011.

\bibitem{danilova.juillard.19}
A.~Danilova and C.~Julliard.
\newblock Understanding volatility, liquidity, and the {T}obin tax.
\newblock {\em Preprint, London School of Economics and Political Science},
  2019.

\bibitem{davila.17}
E.~D{\'a}vila.
\newblock Optimal financial transaction taxes.
\newblock {\em Preprint, New York University}, 2017.

\bibitem{DetempleMurthy.94}
J.~Detemple and S.~Murthy.
\newblock Intertemporal asset pricing with heterogeneous beliefs.
\newblock {\em J. Econ. Theory}, 62(2):294--320, 1994.

\bibitem{DuffieGarleanuPedersen.01}
D.~Duffie, N.~Garleanu, and L.~Pedersen.
\newblock Securities lending, shorting, and pricing.
\newblock {\em J. Finan. Econ.}, 66:307--339, 2002.

\bibitem{EpsteinJi.2013}
L.~Epstein and S.~Ji.
\newblock Ambiguous volatility and asset pricing in continuous time.
\newblock {\em Rev. Financ. Stud.}, 26(7):1740--1786, 2013.

\bibitem{Friedman.64}
A.~Friedman.
\newblock {\em Partial Differential Equations of Parabolic Type}.
\newblock Prentice-Hall, Englewood Cliffs, N.J., 1964.

\bibitem{garleanu.pedersen.13}
N.~Garleanu and L.~H. Pedersen.
\newblock Dynamic trading with predictable returns and transaction costs.
\newblock {\em J. Finance}, 68(6):2309--2340, 2013.

\bibitem{garleanu.pedersen.16}
N.~Garleanu and L.~H. Pedersen.
\newblock Dynamic portfolio choice with frictions.
\newblock {\em J. Econ. Theory}, 165:487--516, 2016.

\bibitem{gonon.al.19}
L.~Gonon, J.~Muhle-Karbe, and X.~Shi.
\newblock Asset pricing with general transaction costs: Theory and numerics.
\newblock {\em Preprint arXiv:1905.05027}, 2019.

\bibitem{HarrisonKreps.78}
M.~Harrison and D.~Kreps.
\newblock Speculative investor behavior in a stock market with heterogeneous
  expectations.
\newblock {\em Quart. J. Econ.}, 92:323--336, 1978.

\bibitem{hau.06}
H.~Hau.
\newblock The role of transaction costs for financial volatility: Evidence from
  the {Paris} bourse.
\newblock {\em J. Eur. Econ. Assoc}, 4(4):862--890, 2006.

\bibitem{HeathSchweizer.00}
D.~Heath and M.~Schweizer.
\newblock Martingales versus {PDE}s in finance: an equivalence result with
  examples.
\newblock {\em J. Appl. Probab.}, 37(4):947--957, 2000.

\bibitem{heaton.lucas.96}
J.~Heaton and D.~J. Lucas.
\newblock Evaluating the effects of incomplete markets on risk sharing and
  asset pricing.
\newblock {\em J. Polit. Economy}, 104(3):443--487, 1996.

\bibitem{herdegen.al.19}
M.~Herdegen, J.~Muhle-Karbe, and D.~Possama{\"i}.
\newblock Equilibrium asset pricing with transaction costs.
\newblock {\em Preprint arXiv:1901.10989v1}, 2019.

\bibitem{jones.seguin.97}
C.~M. Jones and P.~J. Seguin.
\newblock Transaction costs and price volatility: Evidence from commission
  deregulation.
\newblock {\em Am. Econ. Rev.}, 4(87):728--737, 1997.

\bibitem{KaratzasShreve.91}
I.~Karatzas and S.~E. Shreve.
\newblock {\em Brownian Motion and Stochastic Calculus}.
\newblock Springer, New York, 2nd edition, 1991.

\bibitem{ElKarouiPengQuenez.97}
N.~El Karoui, S.~Peng, and M.~C. Quenez.
\newblock Backward stochastic differential equations in finance.
\newblock {\em Math. Finance}, 7(1):1--71, 1997.

\bibitem{kohlmann.tang.02}
M.~Kohlmann and S.~Tang.
\newblock Global adapted solution of one-dimensional backward stochastic
  {Riccati} equations, with application to the mean--variance hedging.
\newblock {\em Stochastic Process. Appl.}, 97(2):255--288, 2002.

\bibitem{Krylov.96}
N.~V. Krylov.
\newblock {\em Lectures on elliptic and parabolic equations in {H}{\"{o}}lder
  spaces}, volume~12 of {\em Graduate Studies in Mathematics}.
\newblock American Mathematical Society, Providence, RI, 1996.

\bibitem{lo.al.04}
A.~W. Lo, H.~Mamaysky, and J.~Wang.
\newblock Asset prices and trading volume under fixed transaction costs.
\newblock {\em J. Polit. Economy.}, 112(5):1054--1090, 2004.

\bibitem{moreau.al.17}
L.~Moreau, J.~Muhle-Karbe, and H.~M. Soner.
\newblock Trading with small price impact.
\newblock {\em Math. Finance}, 27(2):350--400, 2017.

\bibitem{MuhleKarbeNutz.16}
J.~Muhle-Karbe and M.~Nutz.
\newblock A risk-neutral equilibrium leading to uncertain volatility pricing.
\newblock {\em Finance Stoch.}, 22(2):281--295, 2018.

\bibitem{NutzScheinkman.17}
M.~Nutz and J.~A. Scheinkman.
\newblock Shorting in speculative markets.
\newblock {\em To appear in J. Finance}, 2017.

\bibitem{sannikov.skrzypacz.16}
Y.~Sannikov and A.~Skrzypacz.
\newblock Dynamic trading: price inertia and front-running.
\newblock {\em Preprint SSRN:2882809}, 2016.

\bibitem{ScheinkmanXiong.03}
J.~Scheinkman and W.~Xiong.
\newblock Overconfidence and speculative bubbles.
\newblock {\em J. Polit. Economy}, 111:1183--1219, 2003.

\bibitem{ScheinkmanXiong.04}
J.~Scheinkman and W.~Xiong.
\newblock Heterogeneous beliefs, speculation and trading in financial markets.
\newblock In {\em Paris-{P}rinceton {L}ectures on {M}athematical {F}inance
  2003}, volume 1847 of {\em Lecture Notes in Math.}, pages 217--250. Springer,
  Berlin, 2004.

\bibitem{StroockVaradhan.72}
D.~W. Stroock and S.~R.~S. Varadhan.
\newblock On the support of diffusion processes with applications to the strong
  maximum principle.
\newblock In {\em Proceedings of the {S}ixth {B}erkeley {S}ymposium on
  {M}athematical {S}tatistics and {P}robability ({U}niv. {C}alifornia,
  {B}erkeley, {C}alif., 1970/1971), {V}ol. {III}: {P}robability theory}, pages
  333--359. Univ. California Press, Berkeley, Calif., 1972.

\bibitem{Touzi.13}
N.~Touzi.
\newblock {\em Optimal stochastic control, stochastic target problems, and
  backward {SDE}}, volume~29 of {\em Fields Institute Monographs}.
\newblock Springer, New York, 2013.

\bibitem{umlauf.93}
S.~R. Umlauf.
\newblock Transaction taxes and the behavior of the {S}wedish stock market.
\newblock {\em J. Financ. Econ.}, 2(33):227--240, 1993.

\bibitem{vayanos.98}
D.~Vayanos.
\newblock Transaction costs and asset prices: A dynamic equilibrium model.
\newblock {\em Rev. Financ. Stud.}, 11(1):1--58, 1998.

\bibitem{vayanos.vila.99}
D.~Vayanos and J.-L. Vila.
\newblock Equilibrium interest rate and liquidity premium with transaction
  costs.
\newblock {\em Econ. Theory}, 13(3):509--539, 1999.

\bibitem{weston.17}
K.~Weston.
\newblock Existence of a {R}adner equilibrium in a model with transaction
  costs.
\newblock {\em Math. Financ. Econ.}, 12(4):517--539, 2018.

\end{thebibliography}
\bibliographystyle{plain}

\end{document}